\documentclass[a4paper,USenglish,cleveref, autoref, thm-restate]{lipics-v2021}

\title{A Term-based Approach for Generating Finite Automata from Interaction Diagrams}

\author{Erwan Mahe}{Université Paris-Saclay, CEA, List, F-91120, Palaiseau, France}{erwan.mahe@cea.fr}{https://orcid.org/0000-0002-5322-4337}{}

\author{Boutheina Bannour}{Université Paris-Saclay, CEA, List, F-91120, Palaiseau, France}{boutheina.bannour@cea.fr}{https://orcid.org/0000-0002-4943-7807}{}

\author{Christophe Gaston}{Université Paris-Saclay, CEA, List, F-91120, Palaiseau, France}{christophe.gaston@cea.fr}{https://orcid.org/0000-0001-6865-5108}{}

\author{Arnault Lapitre}{Université Paris-Saclay, CEA, List, F-91120, Palaiseau, France}{arnault.lapitre@cea.fr}{https://orcid.org/0000-0002-2185-4051}{}

\author{Pascale Le Gall}{Université Paris-Saclay, CentraleSupélec, F-91192, Gif-sur-Yvette, France}{pascale.legall@centralesupelec.fr}{https://orcid.org/0000-0002-8955-6835}{}

\authorrunning{E. Mahe et al.} 

\Copyright{E. Mahe et al.} 



\relatedversion{} 

\nolinenumbers 

\EventEditors{}
\EventNoEds{1}
\EventLongTitle{}
\EventShortTitle{}
\EventAcronym{}
\EventYear{}
\EventDate{}
\EventLocation{}
\EventLogo{}
\SeriesVolume{}
\ArticleNo{}

\usepackage[font=small]{caption}
\usepackage{subcaption}
\usepackage{wrapfig}

\usepackage{tikz}
\usetikzlibrary{positioning,arrows,shapes,hobby,backgrounds,calc,trees,fit,shapes.multipart,decorations.pathreplacing,automata,positioning}

\usepackage{bussproofs}

\usepackage{fontawesome5}

\usepackage{makecell}

\usepackage{scalerel}
\usepackage{centernot}
\usepackage{mathtools}
\usepackage{bm}

\definecolor{darkspringgreen}{rgb}{0.09, 0.45, 0.27}

\definecolor{hibou_col_lf}{RGB}{22, 22, 130}
\newcommand{\hlf}[1]{\textcolor{hibou_col_lf}{#1}}

\definecolor{hibou_col_ms}{RGB}{15, 86, 15}
\newcommand{\hms}[1]{\textcolor{hibou_col_ms}{#1}}

\newcommand{\shortColRed}[1]{\textcolor{red}{#1}}
\newcommand{\shortColGreen}[1]{\textcolor{darkspringgreen}{#1}}

\newcommand{\shortColOrange}[1]{\textcolor{orange}{#1}}
\newcommand{\shortColViolet}[1]{\textcolor{violet}{#1}}

\DeclareRobustCommand\doubleVerticalTimesDefault{%
  \leavevmode
  {\sbox0{\ddag}%
   \ooalign{\raisebox{\ht0-\height}{$\times$}\cr
            \raisebox{\depth-\dp0}{\scalebox{1}[-1]{$\times$}}\cr}%
  }%
}

\newcommand{\doubleVerticalTimes}{\scalerel*{\doubleVerticalTimesDefault}{b}}

\newcommand{\evadesLfsBase}{\downarrow}
\newcommand{\evadesLfs}[1]{\downarrow^{#1}}

\newcommand{\isPruneBase}{\mathrlap{\raisebox{-.125\height}{\doubleVerticalTimes}}\rightarrow}

\newcommand{\isPruneOf}[1]{\mathrlap{\raisebox{-.125\height}{\doubleVerticalTimes}}\xrightarrow{#1}}

\newcommand{\reachableInts}{\mathtt{reach}}
\newcommand{\nfa}{\mathtt{nfa}}
\newcommand{\nfacompo}{\mathtt{compo}}

\begin{document}

\maketitle

\begin{abstract}
Non-deterministic Finite Automata (NFA) represent regular languages concisely, increasing their appeal for applications such as word recognition.
This paper proposes a new approach to generate NFA from an interaction language such as UML Sequence Diagrams or Message Sequence Charts. Via an operational semantics, we generate a NFA from a set of interactions reachable using the associated execution relation.
In addition, by applying simplifications on reachable interactions to merge them, it is possible to obtain reduced NFA without relying on costly NFA reduction techniques. Experimental results regarding NFA generation and their application in trace analysis are also presented.
\end{abstract}

\section{Introduction}

Interactions are behavioral models describing communication flows between actors.
Interaction languages include, among others, Message Sequence Charts (MSC) \cite{a_theory_of_regular_msc_languages,operational_semantics_for_msc}, or UML Sequence Diagrams (UML-SD) \cite{the_many_meanings_of_uml2_sd_a_survey}. Their main advantage is their easy-to-read graphical representation.
Let us describe their main elements using the sequence diagram drawn in the top left rectangle of Fig.\ref{fig:coreg_nfa_example}. Each actor (here $\hlf{l_1}$, $\hlf{l_2}$ and $\hlf{l_3}$) is associated to a vertical line, called a {\em lifeline}.
Behaviors are described as {\em traces} i.e. successions of atomic communication actions (abbrv. {\em actions}) which are either the emission of a message $\hms{m}$ from a lifeline $\hlf{l}$ or the reception of $\hms{m}$ by $\hlf{l}$.
In the diagrammatic representation, the emission (resp. reception) of a message is represented by an arrow exiting (resp. entering) a lifeline. 
By extension, message exchanges (i.e. an emission and a corresponding reception) are represented by contiguous horizontal arrows connecting two lifelines.
The name of the exchanged message is represented above the corresponding arrow (e.g. $\hms{m_1}$, $\hms{m_2}$ and $\hms{m_3}$ in our example).

Using an asynchronous interpretation of exchanges, these arrows impose a causal order which is that the emission must occur before the corresponding reception (on our example, $\hlf{l_1}$ must emit $\hms{m_1}$ before $\hlf{l_2}$ can receive it).
High-level operators such as various kinds of sequencing, parallel composition, choice and repetition can be used to schedule these atomic exchanges into more complex structured scenarios.
A particularity of interactions is the distinction between strict and weak sequencing. 
While the former always enforces an order between actions, the latter only does so between actions occurring on the same lifeline.
Concurrent-regions (abbrv. {\em co-regions}) further improve expressiveness with the ability to detail which lifelines allow interleavings and which ones do not. In practice co-regions behave like parallel composition on certain lifelines and like weak sequencing on the others.  

Most operators are drawn as annotated boxes (for instance, $loop_S$ denotes a strictly sequential loop in our example).
This is not the case for the weak sequencing operator, which corresponds to the top-to-bottom direction of the diagram and for concurrent regions (co-regions) which are represented using brackets on specific lifelines (e.g. on $\hlf{l_2}$ in our example).
Our example illustrates weak sequencing in two manners.
If one considers the order of the emission of $\hms{m_2}$ w.r.t. that of $\hms{m_3}$, because they both occur on $\hlf{l_3}$, weak sequencing forces $\hlf{l_3}$ to emit $\hms{m_2}$ before it can emit any instance of $\hms{m_3}$.
By contrast, if one considers the emission of $\hms{m_1}$ and that of $\hms{m_2}$, because they resp. occur on $\hlf{l_1}$ and $\hlf{l_3}$, weak sequencing allows them to occur in any order.
In our example, the loop allows arbitrarily many instances of the reception of $\hms{m_3}$ to occur sequentially.
The co-region on $\hlf{l_2}$ then allows the reception of $\hms{m_1}$ to occur before, in-between or after any of these instances of the reception of $\hms{m_3}$.

%

\begin{wrapfigure}{l}{0pt}
    \centering
    \scalebox{.8}{\input{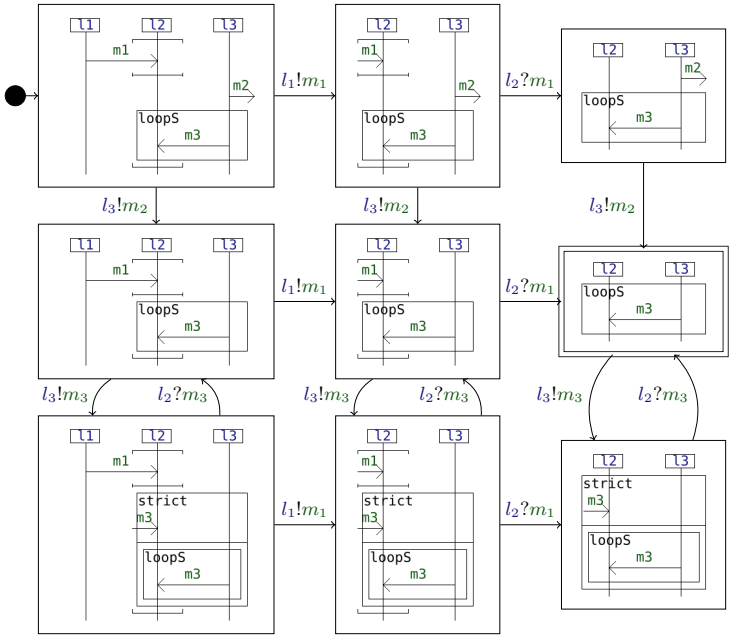}}
    \caption{NFA obtained from $i_0$ drawn in the top left state}
    \label{fig:coreg_nfa_example}
\end{wrapfigure}

Interestingly, for a subset of interactions, associated trace languages are regular. Therefore, associating them with Finite Automata (FA) is possible.
In the literature, translations from interaction to FA generally rely on linearization of partial orders \cite{model_checking_of_message_sequence_charts} and composition, matching interaction operators to FA ones \cite{runtime_monitoring_of_web_service_conversations}.
In this paper, we explore an alternative approach, taking advantage of a representation of interactions as terms \cite{equivalence_of_denotational_and_operational_semantics_for_interaction_languages} in a process algebraic style. 
Using this representation we can build in the manner of \cite{partial_derivatives_of_regular_expressions_and_finite_automata_contructions,generating_optimal_monitors_for_extended_regular_expressions} a NFA in which each state corresponds to an equivalence class of such terms.
Fig.\ref{fig:coreg_nfa_example} illustrates the approach by representing a NFA obtained in this manner.
Each of the $9$ states is associated with a distinct interaction and is represented by a rectangle containing the diagrammatic representation of that interaction.
The initial state is the target of the transition exiting $\bullet$, and rectangles with a doubled border correspond to accepting states. 
Transitions are labelled by actions denoted as either $\hlf{l}!\hms{m}$ or $\hlf{l}?\hms{m}$ for resp. the emission or the reception of message $\hms{m}$ occurring on lifeline $\hlf{l}$.
These actions are immediately executable from the interaction in the source state while the one in the target state specifies continuations of traces specified by the former.

This approach is novel and has three advantages over the literature: \textbf{(1)} the translation of a more expressive interaction language, \textbf{(2)} the ability to directly generate NFA with few states without relying on state reduction algorithms and \textbf{(3)} traceability in monitoring because a deviation from the NFA can be uniquely traced-back to an instant in the unfolding of a global scenario represented by an interaction.

The paper is organized as follows.
After some preliminaries in Sec.\ref{sec:prelim}, we introduce our interaction language and its semantics in Sec.\ref{sec:interaction_language}. 
Sec.\ref{sec:gen_nfa} covers our approach for NFA generation and Sec.\ref{sec:experiments} presents experimental results.
Finally, after some remarks and related works in Sec.\ref{sec:related_works}, we conclude in Sec.\ref{sec:conclusion}.

\section{Preliminaries\label{sec:prelim}}

For a set $A$, $\mathcal{P}(A)$ and $A^*$ resp. denote the sets of all subsets of $A$ and of all finite sequences over $A$.
"$.$" denotes the concatenation operation and $\varepsilon$ the empty sequence.
For $w \in A^*$ and for $a \in A$, we resp. denote by $|w|$ and $|w|_a$ the length of $w$ and the number of occurrences of $a$ in $w$.
A word $u$ is said to be a prefix of $w$ if there exists a word $v$ such that $w =u.v$.

Given two sets $X$ and $Y$, $X \times Y$ denotes their Cartesian product and $Y^X$ denotes the set of all functions $\phi : X \rightarrow Y$ of domain $X$ and codomain $Y$.

\textbf{Non-deterministic Finite Automata (NFA).} Def.\ref{def:epsnfa} recalls the definition of NFA.  
A transition $(q,x,q')$, with $q$ and $q'$ states in $Q$ and $x \in A$ is denoted as $q \overset{x}{\rightsquigarrow} q'$ and is said to be labelled by $x$.
A path from $q$ to $q'$ is a finite sequence of $k > 0$ consecutive transitions $q_i \overset{x_{i+1}}{\rightsquigarrow} q_{i+1}$ with $i \in [0,k-1]$ s.t. $q_0 = q$ and $q_k = q'$. If we denote by $w$ the word obtained from concatenating their labels i.e. $w = x_1.\cdots.x_k$, we may then write $q \overset{w}{\rightsquigarrow} q'$. By extension, we have $q \overset{\varepsilon}{\rightsquigarrow} q$ for any $q \in Q$.
The language recognized by a NFA $\mathcal{A} =(A,Q,q_0,F,\rightsquigarrow)$ is then the set of words $\mathcal{L}(\mathcal{A}) = \{ w \in A^* ~|~ \exists~q_f \in F ~s.t.~ q_0 \overset{w}{\rightsquigarrow} q_f \}$. 

\begin{definition}
\label{def:epsnfa}
A NFA is a tuple $(A,Q,q_0,F,\rightsquigarrow)$ s.t. $A$ is a finite alphabet, $Q$ is a finite set of states, $q_0 \in Q$ is an initial state, $F \subseteq Q$ is a set of accepting states and $\rightsquigarrow \subseteq Q \times A \times Q$ is a finite set of transitions.
\end{definition}

\textbf{Term rewriting.}
A set of operation symbols $\mathcal{F}$ is a structured set $\mathcal{F} = \bigcup_{\substack{j \geq 0}} \mathcal{F}_j$ s.t. for any integer $j \geq 0$, the set $\mathcal{F}_j$ is that of symbols of arity~$j$. Symbols of arity $0$ are constants.
For any set of variables $\mathcal{X}$, we denote by $\mathcal{T}_\mathcal{F}(\mathcal{X})$ the set of terms over $\mathcal{X}$ defined as the smallest set s.t. \textbf{(1)} $\mathcal{F}_0 \cup \mathcal{X} \subset \mathcal{T}_{\mathcal{F}}(\mathcal{X})$ and \textbf{(2)} for any symbol $f \in \mathcal{F}_j$ of arity $j>0$ and for any terms $t_1,\cdots,t_j$ from $\mathcal{T}_{\mathcal{F}}(\mathcal{X})$, $f(t_1,\cdots,t_j) \in \mathcal{T}_{\mathcal{F}}(\mathcal{X})$. 
$\mathcal{T}_\mathcal{F} = \mathcal{T}_\mathcal{F}(\emptyset)$ denotes the set of ground terms.
A function $\phi \in \mathcal{T}_\mathcal{F}^\mathcal{X}$ is extended as a substitution $\phi \in \mathcal{T}_\mathcal{F}^{\mathcal{T}_\mathcal{F}(\mathcal{X})}$ s.t. $\forall~t \in \mathcal{F}_0$, $\phi(t) = t$ 
and for all term the form $f(t_1,\cdots,t_j)$ with $f \in \mathcal{F}_j$,
$\phi(f(t_1,\cdots,t_j)) = f(\phi(t_1),\cdots,\phi(t_n))$.

For any $t \in \mathcal{T}_\mathcal{F}$, we denote by $pos(t) \in \mathcal{P}(\mathbb{N}^*)$ its set of positions \cite{dershowitz_rewrite_systems} which is s.t.
$\forall~t \in \mathcal{F}$ of the form $f(t_1,\cdots,t_j)$ with $f \in \mathcal{F}_j$ we have $pos(f(t_1,\cdots,t_j)) = \{ \varepsilon \} \cup \bigcup_{k \in [1,j]} \{k.p ~|~ p \in pos(t_k)\}$.
For any terms $t$ and $s$ and any position $p \in pos(t)$, $t_{|p}$ denotes the sub-term of $t$ at position $p$ while $t[s]_p$ denotes the term obtained by substituting $t_{|p}$ with $s$ in $t$.

A rewrite rule $x \leadsto y$ is a syntactic entity which relates two terms $x$ and $y$ from $\mathcal{T}_\mathcal{F}(\mathcal{X})$.
A set of rewrite rules $R$ characterizes a Term Rewrite System (TRS) (see Def.\ref{def:trs}) $\rightarrow_R$ which relates ground terms. These terms are s.t. we can obtain the right-hand-side by applying a rewrite rule modulo a substitution at a specific position in the left-hand-side i.e. we have $t \rightarrow_R t[\phi(y)]_p$ with $x \leadsto y$ in $R$ and $t_{|p} = \phi(x)$.
For instance the integer expression simplification $3 + (5 + 0) \rightarrow_R 3 + 5$ can be obtained using $R = \{x + 0 \leadsto x\}$ on $\mathcal{X} = \{x\}$ by applying $\leadsto$ with $\phi(x) = 5$ at position $2$ within $3 + (5 + 0)$.

\begin{definition}
\label{def:trs}
A set $R$ of rewrite rules over $\mathcal{T}_\mathcal{F}(\mathcal{X})$ defines a TRS as a relation:
\[
\rightarrow_R
=
\left\{
(t,t') \in \mathcal{T}_\mathcal{F} \times \mathcal{T}_\mathcal{F}
~\middle|~
\begin{array}{l}
\exists~p \in pos(t),~\exists~(x,y) \in R,~\exists~\phi \in \mathcal{T}_\mathcal{F}^{\mathcal{T}_\mathcal{F}(\mathcal{X})}\\
t_{|p} = \phi(x) \text{ and } t' = t[\phi(y)]_p
\end{array}
\right\}
\]
\end{definition}


Given a TRS $\rightarrow_R$, a term $t \in \mathcal{T}_\mathcal{F}$ is irreducible iff there are no $t' \neq t$ s.t. $t \rightarrow_R t'$.
We denote by $\rightarrow_R^*$ the reflexive and transitive closure of $\rightarrow_R$. A TRS is convergent iff all consecutive applications of $\rightarrow_R$ from a term $t$ necessarily converge to the same irreducible term, i.e. iff $\exists!~t' \in \mathcal{T}_\mathcal{F} ~\text{s.t.}~ t \rightarrow_R^* t'$ and $t'$ is irreducible.
For a convergent TRS $\rightarrow_R$, given a term $t$, the notation $t \rightarrow_R^! t'$ signifies that $t'$ is the unique irreducible term s.t. $t \rightarrow_R^* t'$.

\section{Interactions and their execution\label{sec:interaction_language}}

\subsection{Syntax and principle of the semantics}

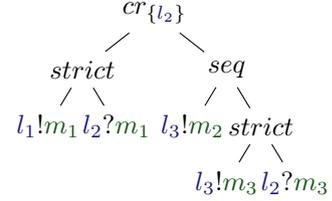
\begin{wrapfigure}{r}{0pt}
    \centering
    \raisebox{0pt}[\dimexpr\height-5\baselineskip\relax]{%
    \scalebox{1}{\begin{tikzpicture}[every node/.style = {shape=rectangle, align=center}]
\node (o) {$cr_{\{\hlf{l_2}\}}$} [level distance=0.75cm,sibling distance=1.9cm]
  child { node (o1) {$strict$} [sibling distance=.9cm]
    child { node (o11) {$\hlf{l_1}!\hms{m_1}$} }
    child { node (o12) {$\hlf{l_2}?\hms{m_1}$} }
  }
  child { node (o2) {$seq$} [sibling distance=.9cm]
    child { node (o21) {$\hlf{l_3}!\hms{m_2}$} }
    child { node (o22) {$strict$} [sibling distance=.9cm]
      child { node (o221) {$\hlf{l_3}!\hms{m_3}$} }
      child { node (o222) {$\hlf{l_2}?\hms{m_3}$} }
    }
  };
\end{tikzpicture}}
    }
    \caption{$i_0$ as a term}
    \label{fig:rv_filrouge_term}
\end{wrapfigure}

The basic notion to describe communications within a concurrent system is that of atomic communication actions that can be observed at the interfaces of its sub-systems.
Those actions correspond to either the emission or the reception of a message.
We use a finite set $L$ of lifelines to describe communication interfaces on which those actions occur (typically one lifeline $l$ for each sub-system).
A finite set $M$ of messages abstracts away all messages that can be transmitted. 
The set of actions is then denoted by $\mathbb{A}=\{ l \Delta m ~|~ l \in L,~\Delta \in \{!,?\},~m \in M \}$. 
Notations $!$ and $?$ resp. denote emissions and receptions.
For any $a \in \mathbb{A}$ of the form $l\Delta m$ with $\Delta \in \{!,?\}$, $\theta(a)$ refers to the lifeline $l$ on which it occurs.
Sequences of actions called {\em traces} characterize executions of systems. 
We denote the set of all traces by $\mathbb{T} = \mathbb{A}^*$.

As done in \cite{equivalence_of_denotational_and_operational_semantics_for_interaction_languages}, we encode interactions as ground terms $\mathcal{T}_\mathcal{F}$ of a language with constants $\mathcal{F}_0$ being either atomic communication actions (elements $a \in \mathbb{A}$) or the empty interaction (denoted as $\varnothing$) which expresses the empty behavior $\varepsilon$. Symbols of arities $1$ and $2$ are then used to encode high level operators.
Def.\ref{def:interaction_language_syntax} formalizes our encoding in the fashion of \cite{equivalence_of_denotational_and_operational_semantics_for_interaction_languages}.
The language includes strict sequencing $strict$, non-deterministic choice $alt$, strictly sequential repetition $loop_S$ and a set of concurrent-region operators $cr_\ell$ one per subset $\ell \subseteq L$ of lifelines.
So as to relate to the more familiar weak sequencing $seq$ and interleaving $par$ operators (see \cite{equivalence_of_denotational_and_operational_semantics_for_interaction_languages}) which are particular cases of co-regions, we denote in the following $cr_\emptyset$ by $seq$ and $cr_L$ by $par$.
Fig.\ref{fig:rv_filrouge_term} illustrates one such interaction term, which encodes as a term the diagrammatic representation of our running example (drawn at the top left of Fig.\ref{fig:coreg_nfa_example}).

\begin{definition}
\label{def:interaction_language_syntax}
The set $\mathbb{I}$ of interactions is the set of ground terms $\mathcal{T}_\mathcal{F}$ built over $\mathcal{F}$ s.t.:
\noindent\[
\mathcal{F}_0 = \mathbb{A} \cup \{\varnothing\}
~~~~~~~
\mathcal{F}_1 = \{loop_S\}
~~~~~~~
\mathcal{F}_2 = \{strict,alt\} \cup \bigcup_{\ell \subset L} \{cr_\ell\}
~~~~~~~
\forall~k > 2, ~\mathcal{F}_k = \emptyset
\]
\end{definition}

In this paper, we define a method to build a NFA such as the one from Fig.\ref{fig:coreg_nfa_example} from a term expressed in the language from Def.\ref{def:interaction_language_syntax} which encodes a sequence diagram.
In order to determine the transitions of that NFA we rely on a relation of the form $i \xrightarrow{a} i'$ where $i$ is an interaction, $a$ is an action which can immediately occur within behaviors specified by $i$ and $i'$ is an interaction specifying all the continuations of behaviors specified by $i$ that start with the occurrence of $a$.
This kind of relation corresponds to execution relations of structural operational semantics \cite{a_structural_approach_to_operational_semantics}.


In \cite{equivalence_of_denotational_and_operational_semantics_for_interaction_languages}, we have defined such a semantics for interactions in the style of Plotkin \cite{a_structural_approach_to_operational_semantics}.
The following introduces a compact rewording of that semantics with two new contributions: \textbf{(1)} the introduction of a delayed-choice rule from \cite{high_level_message_sequence_charts} which allows the generation of smaller NFA and \textbf{(2)} the handling of concurrent regions.

In order to express that semantics we define several predicates inductively on the term structure of interactions. For any $i$ and $i'$ in $\mathbb{I}$, any $\ell \subseteq L$ and any $a \in \mathbb{A}$, these are:
\begin{itemize}
    \item an evasion predicate $i \evadesLfs{\ell}$ which determines whether or not $i$ is able to express traces which involve no action occurring on a lifeline of $\ell$
    \item a pruning relation $i \isPruneOf{\ell} i'$ which characterizes an interaction $i'$ that exactly specifies all the traces expressed by $i$ that involve no action occurring on a lifeline of $\ell$
    \item and finally the non-expression $i \centernot{\xrightarrow{a}}$ and expression $i \xrightarrow{a} i'$ relations.
\end{itemize}

\subsection{Structural operational semantics}

We say that an interaction $i \in \mathbb{I}$ evades a set of lifelines $\ell \subseteq L$ if it can express traces that contain no actions occurring on lifelines of $\ell$. Its definition in Def.\ref{def:evasion} follows directly from the meaning of the involved operators: for example, any loop term $loop_S(i)$ evades any set of lifelines $\ell \subseteq L$ because it specifies the empty behavior $\varepsilon$ which corresponds to repeating 0 times the loop.

\begin{definition}[Evasion]
\label{def:evasion}
The predicate $\evadesLfsBase ~\subseteq \mathbb{I} \times \mathcal{P}(L)$ is s.t. for any $\ell \subseteq L$, any $i_1$ and $i_2$ from $\mathbb{I}$, any $f \in \{strict\} \cup \bigcup_{\ell' \subseteq L} \{cr_{\ell'} \}$ we have:

{
\centering

\noindent\begin{minipage}{1cm}
\begin{prooftree}
\AxiomC{\vphantom{$\theta\evadesLfs{\ell}$}}
\UnaryInfC{$\varnothing \evadesLfs{\ell}$}
\end{prooftree}
\end{minipage}
\begin{minipage}{2cm}
\begin{prooftree}
\AxiomC{\vphantom{$\evadesLfs{\ell}$}$\theta(a) \not\in \ell$}
\UnaryInfC{$a \evadesLfs{\ell}$}
\end{prooftree}
\end{minipage}
\begin{minipage}{4cm}
\begin{prooftree}
\AxiomC{\vphantom{$\theta$}$i_j \evadesLfs{\ell}$}
\RightLabel{$j \in \{1,2\}$}
\UnaryInfC{$alt(i_1,i_2) \evadesLfs{\ell}$}
\end{prooftree}
\end{minipage}
\begin{minipage}{3cm}
\begin{prooftree}
\AxiomC{$i_1 \evadesLfs{\ell}$}
\AxiomC{$i_2 \evadesLfs{\ell}$}
\BinaryInfC{$f(i_1,i_2) \evadesLfs{\ell}$}
\end{prooftree}
\end{minipage}
\begin{minipage}{3cm}
\begin{prooftree}
\AxiomC{\vphantom{$\evadesLfs{\ell}$}}
\UnaryInfC{$loop_S(i_1) \evadesLfs{\ell}$}
\end{prooftree}
\end{minipage}

}
\end{definition}

The pruning relation from Def.\ref{def:pruning_relation} characterizes for any interaction $i$, the existence and unicity of another interaction $i'$ which accepts exactly all executions of $i$ that do not involve any lifelines in $\ell \subseteq L$.

\begin{definition}[Pruning]\label{def:pruning_relation}
The pruning relation $\isPruneBase ~ \subset \mathbb{I} \times \mathcal{P}(L) \times \mathbb{I}$ is s.t. for any $\ell \in L$, any $x \in \{\varnothing\}\cup\mathbb{A}$, any $i_1$ and $i_2$ from $\mathbb{I}$ and any $f \in \{strict,alt\} \cup \bigcup_{\ell' \subseteq L} \{cr_{\ell'} \}$:

{
\centering

\begin{minipage}{3cm}
\begin{prooftree}
\AxiomC{$x \evadesLfs{\ell}$\vphantom{$\isPruneOf{\ell}$}}
\UnaryInfC{$x \isPruneOf{\ell} x$}
\end{prooftree}
\end{minipage}
\begin{minipage}{6cm}
\begin{prooftree}
\AxiomC{$i_j \isPruneOf{\ell} i_j'$\vphantom{$\isPruneOf{\ell}$}}
\AxiomC{$\neg (i_k \evadesLfs{\ell})$}
\RightLabel{
$\{j,k\} = \{1,2\}$
}
\BinaryInfC{$alt(i_1,i_2) \isPruneOf{\ell} i_j'$}
\end{prooftree}
\end{minipage}

\begin{minipage}{4cm}
\begin{prooftree}
\AxiomC{$i_1 \isPruneOf{\ell} i_1'$}
\AxiomC{$i_2 \isPruneOf{\ell} i_2'$}
\BinaryInfC{$f(i_1,i_2) \isPruneOf{\ell} f(i_1',i_2')$}
\end{prooftree}
\end{minipage}
\begin{minipage}{4cm}
\begin{prooftree}
\AxiomC{$i_1 \isPruneOf{\ell} i_1'$}
\UnaryInfC{$loop_S(i_1) \isPruneOf{\ell} loop_S(i_1')$}
\end{prooftree}
\end{minipage}
\begin{minipage}{3.5cm}
\begin{prooftree}
\AxiomC{$\neg(i_1 \evadesLfs{\ell})$\vphantom{$\isPruneOf{\ell}$}}
\UnaryInfC{$loop_S(i_1) \isPruneOf{\ell} \varnothing$}
\end{prooftree}
\end{minipage}

}
\end{definition}

The semantics from \cite{equivalence_of_denotational_and_operational_semantics_for_interaction_languages} relies on the definition of an execution relation which relates interactions to \textbf{(1)} the actions which are immediately expressible and \textbf{(2)} the interactions which remain to be expressed afterwards.
With Def.\ref{def:non_expression}, we can determine which actions cannot be immediately expressed via the $\centernot{\rightarrow}$ relation.

\begin{definition}[Non-expression]
\label{def:non_expression}
The predicate $\centernot{\rightarrow} ~\subseteq \mathbb{I} \times \mathbb{A}$ is s.t. for any $a \in \mathbb{A}$, any $x \in \{\varnothing\}\cup\mathbb{A}$, any $i_1$ and $i_2$ from $\mathbb{I}$ and any $\ell \subseteq L$:

{
\centering

\begin{minipage}{2cm}
\begin{prooftree}
\AxiomC{$x \neq a$\vphantom{$\centernot{\xrightarrow{a}}$}}
\UnaryInfC{$x \centernot{\xrightarrow{a}}$}
\end{prooftree}
\end{minipage}
\begin{minipage}{3cm}
\begin{prooftree}
\AxiomC{$i_1 \centernot{\xrightarrow{a}}$}
\UnaryInfC{$loop_S(i_1) \centernot{\xrightarrow{a}}$}
\end{prooftree}
\end{minipage}
\begin{minipage}{5cm}
\begin{prooftree}
\AxiomC{$i_1 \centernot{\xrightarrow{a}}$}
\AxiomC{$(\neg(i_1 \evadesLfs{L}))\vee(i_2 \centernot{\xrightarrow{a}})$}
\BinaryInfC{$strict(i_1,i_2) \centernot{\xrightarrow{a}}$}
\end{prooftree}
\end{minipage}

\begin{minipage}{6cm}
\begin{prooftree}
\AxiomC{$i_1 \centernot{\xrightarrow{a}}$}
\AxiomC{$(\neg(i_1 \evadesLfs{\{\theta(a)\} \setminus \ell}))\vee(i_2 \centernot{\xrightarrow{a}})$}
\BinaryInfC{$cr_\ell(i_1,i_2) \centernot{\xrightarrow{a}}$}
\end{prooftree}
\end{minipage}
\begin{minipage}{4cm}
\begin{prooftree}
\AxiomC{$i_1 \centernot{\xrightarrow{a}}$}
\AxiomC{$i_2 \centernot{\xrightarrow{a}}$}
\BinaryInfC{$alt(i_1,i_2) \centernot{\xrightarrow{a}}$}
\end{prooftree}
\end{minipage}

}
\end{definition}

The relation $\rightarrow$ given in Def.\ref{def:expression} uses the predicates from Def.\ref{def:evasion}, Def.\ref{def:pruning_relation} and Def.\ref{def:non_expression} to characterize transformations of the form $i \xrightarrow{a} i'$ where $i \in \mathbb{I}$ is an interaction, $a \in \mathbb{A}$ is an action which is immediately executable from $i$ and $i' \in \mathbb{I}$ is a follow-up interaction which characterizes continuations of behaviors specified by $i$ that start with $a$.

\begin{definition}[Expression]
\label{def:expression}
The predicate $\rightarrow ~\subseteq \mathbb{I} \times \mathbb{A} \times \mathbb{I}$ is s.t. for any $a \in \mathbb{A}$, any $i_1,i_2,i_1'$ and $i_2'$ from $\mathbb{I}$, any $f \in \{strict\} \cup \bigcup_{\ell' \subseteq L} \{cr_{\ell'} \}$ and any $\{j,k\} = \{1,2\}$:

{
\centering

\begin{minipage}{2cm}
\begin{prooftree}
\AxiomC{\phantom{$\xrightarrow{a}$}}
\RightLabel{act}
\UnaryInfC{$a \xrightarrow{a} \varnothing$}
\end{prooftree}
\end{minipage}
\begin{minipage}{7cm}
\begin{prooftree}
\AxiomC{$i_1 \xrightarrow{a} i_1'$}
\RightLabel{loop}
\UnaryInfC{$loop_S(i_1) \xrightarrow{a} strict(i_1',loop_S(i_1))$}
\end{prooftree}
\end{minipage}
\begin{minipage}{4.75cm}
\begin{prooftree}
\AxiomC{$i_1 \xrightarrow{a} i'_1$}
\RightLabel{f-left}
\UnaryInfC{$f(i_1,i_2) \xrightarrow{a} f(i'_1,i_2)$}
\end{prooftree}
\end{minipage}

\vspace*{.1cm}

\begin{minipage}{6.5cm}
\begin{prooftree}
\AxiomC{$i_2 \xrightarrow{a} i'_2$}
\AxiomC{$i_1 \evadesLfs{L}$}
\RightLabel{strict-right}
\BinaryInfC{$strict(i_1,i_2) \xrightarrow{a} i'_2$}
\end{prooftree}
\end{minipage}
\begin{minipage}{6.5cm}
\begin{prooftree}
\AxiomC{$i_2 \xrightarrow{a} i_2'$}
\AxiomC{$i_1 \isPruneOf{\{\theta(a)\} \setminus \ell} i_1'$}
\RightLabel{cr-right}
\BinaryInfC{$cr_\ell(i_1,i_2) \xrightarrow{a} cr_\ell(i_1',i_2')$}
\end{prooftree}
\end{minipage}

\vspace*{.1cm}

\begin{minipage}{6.5cm}
\begin{prooftree}
\AxiomC{$i_j \xrightarrow{a} i'_j$}
\AxiomC{$i_k \centernot{\xrightarrow{a}}$}
\RightLabel{alt-choice}
\BinaryInfC{$alt(i_1,i_2) \xrightarrow{a} i'_j$}
\end{prooftree}
\end{minipage}
\begin{minipage}{6.5cm}
\begin{prooftree}
\AxiomC{$i_1 \xrightarrow{a} i'_1$}
\AxiomC{$i_2 \xrightarrow{a} i'_2$}
\RightLabel{alt-delay}
\BinaryInfC{$alt(i_1,i_2) \xrightarrow{a} alt(i'_1,i'_2)$}
\end{prooftree}
\end{minipage}

}

\end{definition}

The relation from Def.\ref{def:expression} resembles those found in process algebra.
Rules ``\textit{act}'', ``\textit{loop}'' and ``\textit{f-left}'' are self-explanatory.
Given $i_1$ and $i_2$ two interactions, we can execute an action $a_2$ on the right of either $strict(i_1,i_2)$ or $cr_\ell(i_1,i_2)$ (with $\ell \subseteq L$) i.e. resulting from the execution of $i_2$ iff it does not contradict the partial orders imposed by either $strict$ or $cr_\ell$ between the actions of $i_1$ and those of $i_2$.

Rule ``\textit{strict-right}'' imposes the termination of $i_1$ before any action from $i_2$ can be executed.
This termination is possible iff $i_1$ can express the empty behavior $\varepsilon$ which is characterized by $i_1 \evadesLfs{L}$ because $\varepsilon$ is the only behavior which involves no action occurring on any one of the lifelines of $L$.
If $a_2$ is executed, we then consider that $i_1$ has terminated (otherwise we might subsequently observe actions $a_1$ from $i_1$ which contradicts the order imposed by $strict$) and there only remains to execute $i_2'$ which is s.t. $i_2 \xrightarrow{a} i_2'$.

If $\ell = L$, we are in the presence of interleaving operator $par = cr_L$. 
In that case, it is always possible to execute $a_2$ from $i_2$ because any action $a_1$ from $i_1$ can occur either before or after $a_2$.
This is reflected by the fact that $\{\theta(a)\} \setminus L = \emptyset$ and that $i_1 \isPruneOf{\emptyset} i_1$ always holds.
Hence we have $i'_1 = i_1$ and the predicate $i_1 \isPruneOf{\emptyset} i_1 = \top$ which makes rule ``\textit{cr-right}'' coincide with the classical right-rule for interleaving \cite{equivalence_of_denotational_and_operational_semantics_for_interaction_languages}.
    
If $\ell = \emptyset$, we are in the presence of weak sequencing operator $seq = cr_\emptyset$. 
In that case let us consider an action $a_1$ from $i_1$ occurring on the lifeline $\theta(a_2)$ on which $a_2$ occurs.
If $a_1$ must occur whenever $i_1$ is executed, then it must logically occur before $a_2$ (as imposed by $seq$).
Hence, if $a_2$ occurs first, this means $a_1$ must not occur at all.
It is possible for $i_1$ not to express $a_1$ (and any other action occurring on $\theta(a_2)$) iff $i_1 \evadesLfs{\{\theta(a)\}}$.
In that case, in order to compute a follow-up to the execution of $a_2$ from $i_2$ in $seq(i_1,i_2)$, we need to clean-up $i_1$ from any action occurring on $\theta(a_2)$. This is the role of the pruning predicate $\isPruneBase$ which intervenes in ``\textit{cr-right}'' via the condition of the existence of a $i_1'$ s.t. $i_1 \isPruneOf{\{\theta(a)\}} i_1'$.
In the pruned interaction $i_1'$, we only preserve behaviors of $i_1$ that do not contradict the fact that $a_2$ occurs before any action from $i_1$.
In that sense, and because $\ell = \emptyset$, the ``\textit{cr-right}'' rule coincides with the classical right-rule for weak sequencing \cite{equivalence_of_denotational_and_operational_semantics_for_interaction_languages}.
   
For all the other cases i.e. whenever $\ell \subsetneq L$ and $\ell \neq \emptyset$, the condition $i_1 \isPruneOf{\{\theta(a)\} \setminus \ell} i_1'$ of rule ``\textit{cr-right}'' makes so that $cr_\ell$ behaves like interleaving for actions occurring on lifelines of $\ell$ and like weak sequencing for those occurring on $L \setminus \ell$.
Fig.\ref{fig:sem_coreg} illustrates the use of a co-region and of pruning on an example. 
Here $cr_{\{\hlf{l_1}\}}$ makes so that the emission of $\hms{m_1}$ and $\hms{m_2}$ can occur in any order but that $\hms{m_1}$ must be received before $\hms{m_2}$ if it is ever received.
When we execute $\hlf{l_1}!\hms{m_2}$, because it can occur after $\hlf{l_1}!\hms{m_1}$ due to the concurrent region on $\hlf{l_1}$, the pruning predicate $\isPruneOf{\hlf{l_1}}$ does not force the choice of a branch of the alternative.
However, when we execute $\hlf{l_2}?\hms{m_2}$, pruning forces the right branch of the alternative (containing the empty interaction $\varnothing$) to be chosen. Otherwise, we would risk having $\hlf{l_2}?\hms{m_1}$ occur afterwards, which is forbidden by the weak sequencing.

\begin{remark}
\label{rem-exec-par}
The $par = cr_L$ operator is the most permissive scheduling operator (among $strict$ and all the $cr_\ell$ with $\ell \subseteq L$ which includes $seq = cr_\emptyset$). 
Indeed, all left rules have the same form and thus allow the same derivations while the right rules contain restrictive conditions for $strict$ and all $cr_\ell$ (resp. $i_1 \evadesLfs{L}$ for $strict$ and $i_1 \isPruneOf{\{\theta(a)\} \setminus \ell} i_1'$ for $cr_\ell$ with $\ell \subsetneq L$) but not for $par = cr_L$ because $\{\theta(a)\} \setminus L = \emptyset$ and we always have that $i_1 \isPruneOf{\emptyset} i_1$ holds.
As a result, any $f \in \{strict\} \cup \bigcup_{\ell \subseteq L} \{cr_\ell\}$ allows less derivations than $par = cr_L$ by construction.
\end{remark}

\begin{figure}[h]
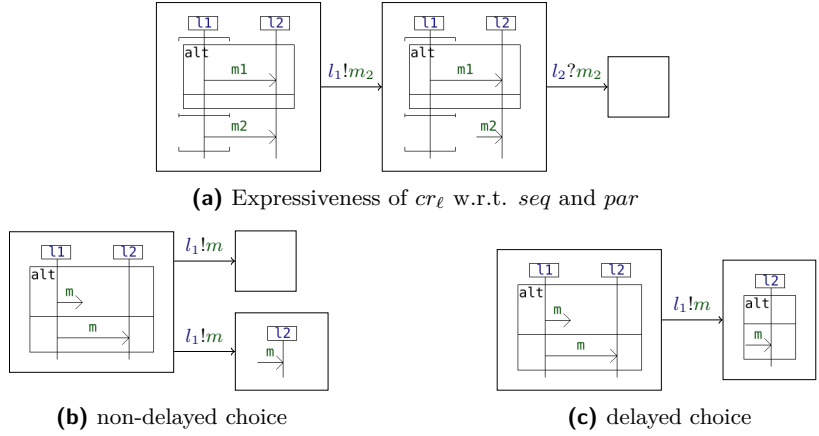

\centering

\begin{subfigure}{\linewidth}
\captionsetup{justification=centering}
    \centering
    \scalebox{.8}{\input{images/diff/coreg_seq/graph}}
    \caption{Expressiveness of $cr_\ell$ w.r.t. $seq$ and $par$}
    \label{fig:sem_coreg}
\end{subfigure}

\vspace*{.2cm}

\begin{subfigure}{.45\linewidth}
\captionsetup{justification=centering}
    \centering
    \scalebox{.8}{\input{images/alt/normal/graph}}
    \caption{non-delayed choice}
    \label{fig:sem_alt_non_delayed}
\end{subfigure}
\begin{subfigure}{.45\linewidth}
\captionsetup{justification=centering}
    \centering
    \scalebox{.8}{\input{images/alt/delayed/graph}}
    \caption{delayed choice}
    \label{fig:sem_alt_delayed}
\end{subfigure}

    \caption{Examples of interaction execution}
    \label{fig:exec_sem}
\end{figure}

In our formalism strict sequencing $strict$ is used to enforce a strict order between any two actions. 
We use it to encode the asynchronous passing or broadcast of a message in e.g. $strict(\hlf{l_1}!\hms{m},\hlf{l_2}?\hms{m})$ and also as the scheduling operator associated to the strictly sequential loop $loop_S$.
The co-region operators encode both parallel composition and weak sequencing and provides additional expressiveness w.r.t. \cite{equivalence_of_denotational_and_operational_semantics_for_interaction_languages}.
It can be used e.g. for specifying that some messages are emitted in a specific order but may be received in any order.

There are several trace-equivalent (as opposed to bissimilar) manners to define an operational rule for the non-deterministic choice $alt$ operator. Fig.\ref{fig:sem_alt_non_delayed} and Fig.\ref{fig:sem_alt_delayed} present two such manners. In \cite{revisiting_semantics_of_interactions_for_trace_validity_analysis,a_small_step_approach_to_multi_trace_checking_against_interactions}, we have presented the manner described on Fig.\ref{fig:sem_alt_non_delayed} in which choice of alternative branches is made as soon as possible. In this paper, we rather favor the one described on Fig.\ref{fig:sem_alt_delayed} that is called delayed-choice in \cite{high_level_message_sequence_charts} as its use will further reduce the number of states of generated NFA.

\begin{definition}[Semantics]
\label{def:semantics}
$\sigma : \mathbb{I} \rightarrow \mathcal{P}(\mathbb{T})$ is defined by:

{
\centering

\begin{minipage}{4cm}
\begin{prooftree}
\AxiomC{$i \evadesLfs{L}$}
\UnaryInfC{$\varepsilon \in \sigma(i)$}
\end{prooftree}
\end{minipage}
\begin{minipage}{4cm}
\begin{prooftree}
\AxiomC{$t \in \sigma(i')$}
\AxiomC{$i \xrightarrow{a} i'$}
\BinaryInfC{$a.t \in \sigma(i)$}
\end{prooftree}
\end{minipage}

}

\end{definition}

The predicates from Def.\ref{def:evasion} and Def.\ref{def:expression} enable us to define an operational semantics of interactions in Def.\ref{def:semantics}. The set $\sigma(i)$ of traces of an interaction $i$ contains all (possibly empty) sequences $a_1 a_2 \ldots a_k$ of actions such that there exist some execution steps $i \xrightarrow{a_1} i_1$, 
$i_1 \xrightarrow{a_2} i_2$, \ldots $i_{k-1} \xrightarrow{a_k} i_k$ with 
$i_k \evadesLfs{L}$. In particular, if the sequence is empty, then $i \evadesLfs{L}$ and $\varepsilon \in \sigma(i)$.
We denote by $\overset{*}{\rightarrow}$ the reflexive and transitive closure of the execution relation~$\rightarrow$.

The above predicate definitions are consistent in the sense that they are interdependent and complementary. For instance, for any $\ell \subseteq L$, any $i$ in $\mathbb{I}$, $i \evadesLfs{\ell}$ iff there exists a unique $i'$ in $\mathbb{I}$ s.t. $i \isPruneOf{\ell} i'$. Other such properties can be stated about the operational semantics, but are not essential for the rest of the paper. In the sequel we essentially use the rules of Def.\ref{def:expression}.

\section{Deriving NFA from interactions\label{sec:gen_nfa}}

\subsection{Our approach}

\begin{wrapfigure}{r}{0pt}
    \centering
    \raisebox{0pt}[\dimexpr\height-5\baselineskip\relax]{%
\scalebox{.8}{
\begin{tikzpicture}
\node (int) at (-.5,0) {\includegraphics[width=1.5cm]{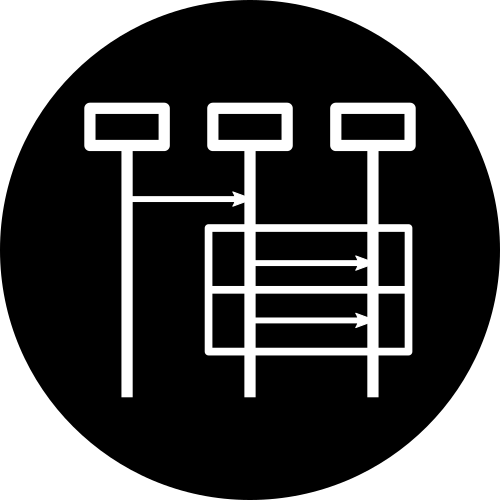}};
\node (nfa) at (3,2.25) {\includegraphics[width=1.5cm]{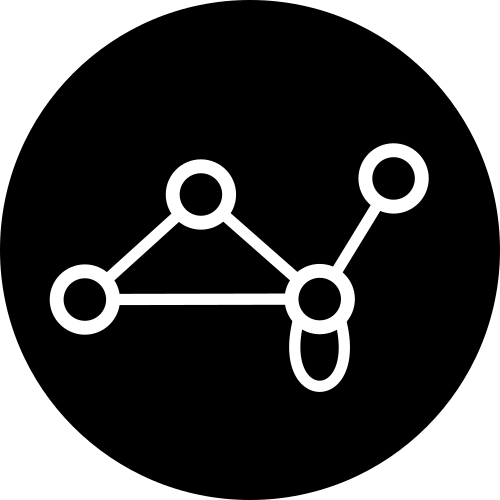}};
\node (nfared) at (3,-2.25) {\includegraphics[width=1.5cm]{images/icons/icon_nfa.png}};
\node (lang) at (2,0) {
\begin{tikzpicture}
\node[inner sep=0,outer sep=0] (t0) at (0,0) {\includegraphics[width=.75cm]{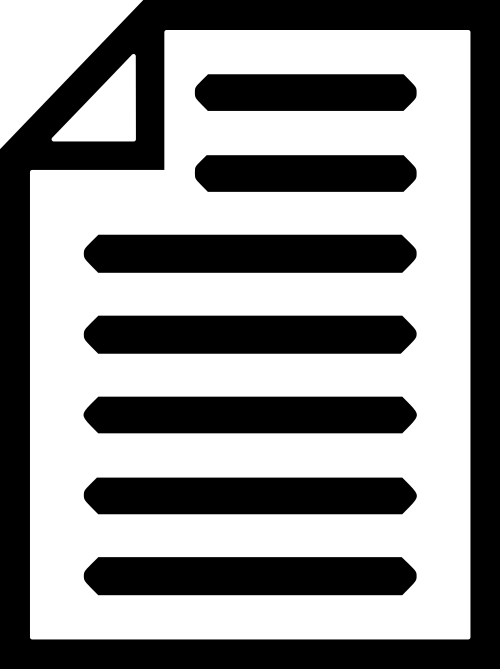}};
\node[inner sep=0,outer sep=0,below left=-.6cm and -.5cm of t0] (t1) {\includegraphics[width=.75cm]{images/icons/icon_text_file_filledwhite.png}};
\node[inner sep=0,outer sep=0,below right=-.6cm and -.5cm of t1] (t2) {\includegraphics[width=.75cm]{images/icons/icon_text_file_filledwhite.png}};
\node[inner sep=0,outer sep=0,above right=-.8cm and -.25cm of t2] (t3) {\includegraphics[width=.75cm]{images/icons/icon_text_file_filledwhite.png}};
\end{tikzpicture}
};
\node[right = -.1cm of lang,align=center] (tl) {trace\\language};
\node[left = -.1cm of int] (il) {\rotatebox{90}{interaction}};
\node[right = 0cm of nfa,align=center] (nl) {NFA};
\node[right = 0cm of nfared,align=center] (rl) {reduced\\NFA};
\draw[->,line width=3pt,red] (int) edge[bend left=30,sloped] node[above,midway,black] {translation} (nfa);
\draw[->,line width=3pt,darkspringgreen] (int) edge[bend right=30,sloped] node[below,midway,black] {translation} (nfared);
\draw[->,line width=3pt,red] (nfa) edge[bend left=60,sloped] node[above,midway,black] {state reduction} (nfared);
\draw[->] (int) -- node[above,midway] {$\mathcal{L}$} (lang);
\draw[->] (nfa) -- node[right,midway] {$\mathcal{L}$} (lang);
\draw[->] (nfared) -- node[below left=-.1cm and .01cm,midway] {$\mathcal{L}$} (lang);
\node (leg1) at (1.1,-3.2) {\shortColRed{{\scriptsize\faSquare}\small~careless translation then reduction}};
\node[below=-.1cm of leg1] (leg2) {\shortColGreen{{\scriptsize\faSquare}\small~translation with care for reduction}};
\end{tikzpicture}
}
    }
    \caption{Translations}
    \label{fig:int_to_nfa_schema}
\end{wrapfigure}

A FA generated from an interaction describing a real-world system may be extremely large (in number of states) which poses problems related to the time for building the FA or the space for storing it. 

State reduction and state minimization correspond to obtaining an equivalent FA with resp. fewer states and as few states as possible.
It is known that a minimal DFA may have an exponentially larger number of states than that of an equivalent minimal NFA \cite{succint_representation_of_regular_languages}. For instance, given alphabet $\{a,b\}$, the regular expression
$(a|b)^*a(a|b)^n$ can be encoded as a NFA with $n+2$ states while an equivalent minimal DFA has $2^{n+1}$ states.
Because interactions can encode such expressions, choosing NFA as a target formalism is preferable if we want to optimise towards state reduction.
Also, operators such as $par = cr_L$, $seq=c_\emptyset$, $alt$ and $loop_S$ make behaviors of interactions non-deterministic and cyclic.

However, the problem of state minimization for NFA is PSPACE-complete \cite{minimal_nfa_problems_are_hard}.
Various algorithms exist, such as Kameda-Weiner \cite{on_the_state_minimization_of_nondeterministic_finite_automata} or \cite{a_new_algorithm_of_the_state_minimization_for_the_nondeterministic_finite_uautomata,minimilizations_of_nfa_using_the_universal_automaton}. 
Still, scalability is problematic, as demonstrated by the search for approximal solutions (i.e. to find reduced but not minimal NFA in a more reasonable time)~\cite{nfa_reduction_algorithms_by_means_of_regular_inequalities,on_nfa_reductions,algorithms_for_computing_small_nfas,reducing_nondeterministic_finite_automata_with_SAT_solvers}.
\cite{algorithms_for_computing_small_nfas} relies on identifying and merging equivalent states while \cite{reducing_nondeterministic_finite_automata_with_SAT_solvers} encodes the problem in SAT.
Yet, experimental validation is limited to small NFA (in \cite{reducing_nondeterministic_finite_automata_with_SAT_solvers} they consider at most $15$ states).
Vastly more efficient algorithms exist for strict subsets of NFA such as DFA (e.g. Hopcroft \cite{hopcroft_dfa_min}) or acyclic automata \cite{incremental_construction_of_minimal_acyclic_finite_state_automata,fa_minimization_heuristics_for_a_class_of_finite_languages} but they cannot be used to minimize FA obtained from interactions (for which the minimal FA is a minimal NFA which can be cyclic and non deterministic).

Fig.\ref{fig:int_to_nfa_schema} describes two manners to generate reduced NFA from interactions.
The \shortColRed{red path} uses of a translation mechanism that is not optimized for state reduction followed by that a NFA state reduction technique.
By contrast, the \shortColGreen{green path} uses a translation mechanism which incorporates on-the-fly state reduction so as to directly obtain a reduced NFA.
Due to the aforementioned difficulty of NFA reduction, we argue that such an approach is preferable.



In the following, we detail one such mechanism.
We prove that the set of interactions reachable by iterating the application of elementary steps $i \xrightarrow{a} i'$ is finite which enables the construction of a NFA in which each state corresponds to an interaction.
We explain how we can reduce the number of states on-the-fly by simplifying interactions as they are discovered (using term rewriting).
Memoization of already encountered simplified terms allows building the set of states of the NFA while the identification of transformations $i \xrightarrow{a} i'$ allows building its set of transitions and a state corresponding to a term $i$ is accepting iff $i \evadesLfs{L}$.

\subsection{Details of the NFA generation mechanism}

Given an interaction $i$, the set of interactions reachable from $i$ with the relation $\rightarrow$ is finite. This allows us to define a finite set $\reachableInts(i)$ of such terms.

\begin{lemma}
\label{lem:finite_reachable}
For any $i \in \mathbb{I}$, the set $\reachableInts(i) = \{ i' ~|~ i \overset{*}{\rightarrow} i' \}$ is finite.
\end{lemma}

\begin{proof}
For any interaction $i \in \mathbb{I}$ let us reason by induction on interactions in $\mathbb{I}$.

For $i = \varnothing$, $\reachableInts(i)$ is reduced to $\{\varnothing\}$ since no rule is applicable.
For $i = a \in \mathbb{A}$, $\reachableInts(i) = \{a,\varnothing\}$ since only rule ``\textit{act}'' is applicable and yields $a \xrightarrow{a} \varnothing$.

Considering $par = cr_L$, for $i=par(i_1,i_2)$, we have $\reachableInts(i) = \{ par(j_1,j_2) ~|~ j_1 \in \reachableInts(i_1), j_2 \in \reachableInts(i_2) \}$. 
Indeed, from terms of this form, we can apply either rule ``\textit{f-left}'' or ``\textit{cr-right}''.
The former replaces the sub-term $i_1$ by a certain $i_1'$ leading to a term of the form $par(i'_1,i_2)$.
Because whichever is the action $a$ that is executed we have $\{\theta(a)\}\setminus L= \emptyset$, we always have $i_1 \isPruneOf{\emptyset} i_1$ as the right-hand-side condition of the application of ``\textit{cr-right}''.
As a result, applying it yields a term of the form $par(i_1,i_2')$, with $i_1$ being preserved and $i_2'$ reached from $i_2$.
The same two rules are again the only ones applicable from these terms so that by applying them several times, successive derivations result in all terms of the form $par(j_1,j_2)$ with $j_1$ (resp $j_2$) reachable from $i_1$ (resp. $i_2$), i.e. $j_1 \in \reachableInts(i_1)$ (resp. $j_2 \in \reachableInts(i_2)$). Both $\reachableInts(i_1)$ and $\reachableInts(i_2)$ are finite by induction and hence $\reachableInts(i)$ is also finite.

For $i=f(i_1,i_2)$ with $f \in \{strict\} \cup \bigcup{\ell \subsetneq L} cr_\ell$, we have $|\reachableInts(i)| \leq |\reachableInts(par(i_1,i_2))|$. 
Indeed, as per Rem.\ref{rem-exec-par}, these operators enable less interleavings than $par$. Derivations from a term of the form $f(i_1,i_2)$ lead to terms of the form $f(i'_1,i'_2)$ with the property that the derivation $par(i_1,i_2) \xrightarrow{a} par(i'_1,i'_2) $ also holds. As by hypothesis (previous point), $\reachableInts(par(i_1,i_2))$ is finite, $\reachableInts(f(i_1,i_2))$ is also finite.  

For $i=alt(i_1,i_2)$, then $\reachableInts(i) \subset \{i\} \cup \reachableInts(i_1) \cup \reachableInts(i_2) \cup \{ alt(j_1,j_2) ~|~ j_1 \in \reachableInts(i_1), j_2 \in \reachableInts(i_2) \}$. 
Indeed, a derivation $alt(i_1,i_2) \xrightarrow{a} i'$ comes from either the rule ``\textit{alt-choice}'' or the rule ``\textit{alt-delay}'':
\begin{itemize}
    \item if it comes from the rule ``\textit{alt-choice}'', then it means that there is a $k \in \{1,2\}$ s.t. $i_k \xrightarrow{a} i'$ and $i' \in \reachableInts(i_k)$ as well as all the following interactions.
    \item if it comes from the rule ``\textit{alt-delay}'', then $i'$ is of the form $alt(i_1',i_2')$ and, by induction, all continuations are in either $\reachableInts(i_1')$, $\reachableInts(i_2')$ or $\{ alt(j_1,j_2) ~|~ j_1 \in \reachableInts(i_1'), j_2 \in \reachableInts(i_2') \}$ 
\end{itemize}
The induction hypothesis implies that both $\reachableInts(i_1)$ and $\reachableInts(i_2)$ are finite. Hence $\reachableInts(i)$ is also finite.

For $i=loop_S(i_1)$, then $\reachableInts(i) = \{i\} \cup \{ strict(j_1,i) ~|~ j_1 \in \reachableInts(i_1) \}$. For such terms, only rule ``\textit{loop}'' is applicable so that the first derivation (if it exists) results in a term of the form $strict(i'_1,loop_S(i_1))$. From there:
\begin{itemize}
    \item if we apply the ``\textit{f-left}'' rule, 
    we obtain a term of the form $strict(i''_1,$ $loop_S(i_1))$ with $i'_1 \xrightarrow{a} i''_1$. With further applications of rule ``\textit{f-left}'', we obtain all terms of the form $strict(j_1,loop_S(i_1))$ with $j_1 \in \reachableInts(i_1)$.
    \item if we apply ``\textit{strict-right}'', it means that $i'_1\!\evadesLfs{L}$ and that we have $loop_S(i_1)\!\xrightarrow{a} strict(i^*_1,loop_S(i_1))$ with $i_1 \xrightarrow{a} i^*_1$ so that we obtain $strict(i'_1,loop_S(i_1))\!\xrightarrow{a} strict(i^*_1,loop_S(i_1))$. As $i^*_1$ belongs to $\reachableInts(i_1)$, by applying the ``\textit{strict-right}'' rule, we find yet again an interaction of the form $strict(j_1,i)$ with $j_1 \in \reachableInts(i_1)$. 
\end{itemize}
Due to the induction hypothesis, $\reachableInts(i_1)$ is finite and as a result, $\reachableInts(i) = \{i\} \cup \{ strict(j_1,i) ~|~ j_1 \in \reachableInts(i_1) \}$ is also finite.
\end{proof}

Lem.\ref{lem:finite_reachable} defines the set $\reachableInts(i)$ of all the interactions that can be reached from $i$ via the execution relation $\rightarrow$ and states that this set is always finite. In Def.\ref{def:translate_from_int_to_nfa}, this enables building a NFA whose recognized language coincides with the set of traces of the interaction $i$.

\begin{definition}
\label{def:translate_from_int_to_nfa}
For any $i_0 \in \mathbb{I}$, $\nfa(i_0) = (\mathbb{A}, Q, q_0, F, \rightsquigarrow)$ is the NFA whose elements are defined by: $Q = \reachableInts(i_0)$, $q_0 = i_0$, $F = \{ i ~|~ i \in Q, i \evadesLfs{L} \}$ and $\rightsquigarrow = \{(i,a,i') ~| (i,i') \in Q^2, a \in  \mathbb{A}, ~ i \xrightarrow{a} i' \}$.
\end{definition}

Def.\ref{def:translate_from_int_to_nfa} is well founded. Indeed, the set $Q$ is finite according to Lem.\ref{lem:finite_reachable}, the subset $F$ of accepting states and the transition relation $\rightsquigarrow$ are defined thanks to predicates $\evadesLfs{L}$ and $\rightarrow$ from Def.\ref{def:evasion} and Def.\ref{def:expression}.

Th.\ref{th:safety_translation} states the soundness of the translation i.e. that the language recognized by $\nfa(i)$ corresponds to the semantics $\sigma(i)$.

\begin{theorem}
\label{th:safety_translation}
For any interaction $i \in \mathbb{I}$, we have 
$
\mathcal{L}(\nfa(i)) = \sigma(i)
$
\end{theorem}

\begin{proof}
The construction of the automaton reflects the construction of the execution tree: there is a transition $i' \overset{a}{\rightsquigarrow} i''$ in $\nfa(i)$ for each execution step $i' \xrightarrow{a} i''$ issued from $i'$ and reciprocally.  
As accepting states of $\nfa(i)$ coincide with reachable interactions $i'$ verifying $i' \evadesLfs{L}$, then sequences recognized by $\nfa(i)$ are traces of $i$, i.e. we have $\mathcal{L}(\nfa(i)) = \sigma(i)$.
\end{proof}

In most cases, the iterative application of the execution relation $\rightarrow$ creates many useless occurrences of $\varnothing$ in the derived interaction terms. 
These sub-terms can be removed without the set of traces being modified. 
For example, because of the rule ``\textit{f-left}'' derivations from $strict(i,\varnothing)$ are exactly those from $i$. $strict(i,\varnothing)$ and $i$ are said to be semantically equivalent.
The construction of the NFA associated with an interaction can take advantage of considering such semantically equivalent terms in order to reduce the number of states given that each state corresponds to a term.
In order to set up these simplification mechanisms, we provide in Def.\ref{def:simp_varnothing} a set $R$ of rewrite rules aimed at eliminating useless occurrences of $\varnothing$.

\begin{definition}
\label{def:simp_varnothing}
Given a variable $x$, $R$ is the following set of rewrite rules over $\mathcal{T}_\mathcal{F}(\{x\})$:
\[
R = 
\left( 
\begin{array}{l}
\cup_{f \in \{strict\} \cup \bigcup_{\ell \subset L} \{cr_\ell\}} ~ \{ f(\varnothing,x) \leadsto x, ~ f(x,\varnothing) \leadsto x \} \\
\cup ~ \{ alt(\varnothing,loop_S(x)) \leadsto loop_S(x),~ alt(loop_S(x),\varnothing) \leadsto loop_S(x) \} \\
\cup ~ \{ alt(\varnothing,\varnothing) \leadsto \varnothing ,~ loop_S(\varnothing) \leadsto \varnothing \} \\
\end{array}
\right)
\]
\end{definition}

These rewrite rules define a TRS $\rightarrow_R$ which is convergent and preserves the semantics $\sigma$ (as per Lem.\ref{lem:intsimpl_prop}).

\begin{lemma}
\label{lem:intsimpl_prop}
The TRS characterized by $R$ is convergent and semantically sound i.e. for any $i \in \mathbb{I}$,~$\exists!~i_s \in \mathbb{I}$ s.t. $i \rightarrow^!_R i_s$ and we have $\sigma(i) = \sigma(i_s)$.
\end{lemma}

\begin{proof}
$\rightarrow_R$ is both terminating and confluent (see automated proof in \cite{hibou_trs_simplify_empty} using the TTT2 \cite{tyrolean_termination_tool_2} and CSI \cite{csi_a_confluence_tool} tools) and hence convergent. 
To prove semantic equivalence it suffices to prove it for interactions related by all $\leadsto \in R$ which is trivial.
\end{proof}

Lem.\ref{lem:intsimpl_prop} stating that the TRS is convergent, for any $i \in \mathbb{I}$, there exists a unique $i_s$ such that $i \rightarrow^!_S i_s$.
With the index $s$ standing for ``simplified'', we use the notation $\__s$ to designate simplified terms. This allows us to consider simplified NFA generation in Def.\ref{def:simpl_translate_from_int_to_nfa}.

\begin{definition}
\label{def:simpl_translate_from_int_to_nfa}
For any $i_0 \in \mathbb{I}$, $\nfa_s(i_0) = (\mathbb{A}, Q, q_0, F, \rightsquigarrow)$ is the NFA s.t.: \\
$Q = \{ i_s \in \mathbb{I} ~|~ i \in \reachableInts(i_0) \}$, $q_0 = i_{0s}$, $F = \{ i_s ~|~ i \in Q, i \evadesLfs{L} \}$
and \\
\centerline{$\rightsquigarrow = \{(i_s,a,i'_s) ~| (i,i') \in reach(i_0)^2, a \in  \mathbb{A}, ~ i \xrightarrow{a} i' \}$}.
\end{definition}

In a few words, $\nfa_s(i_0)$ is the quotient NFA of the NFA $\nfa(i_0)$ by the equivalence relation on interactions defined by: $i \equiv i'$ iff $i_s = i'_s$. In particular, transitions between two simplified interactions $j$ and $j'$ include all transitions between interactions $i$ and $i'$ verifying $i_s = j$ and $i'_s = j'$. 
Those results can be generalized for any convergent and semantically sound Term Rewrite System for interactions\footnote{A more complete sound TRS for interactions, taking advantage of algebraic properties of operators (associativity of $strict$, $seq$, $par$, $alt$, commutativity of $alt$ and $par$, idempotency of loops etc.) is given in \cite{mahe:tel-03369906} and its convergence proven modulo AC.}. In fact, each node of the generated NFA may rather correspond to an equivalence class of (semantically equivalent) interactions according to a certain equivalence relation which preserves semantics $\sigma$.

\begin{figure}[h]
    \centering

\noindent\begin{minipage}{.49\linewidth}
        \centering

\begin{subfigure}{.975\linewidth}
\captionsetup{justification=centering}
\centering
\includegraphics[scale=.3]{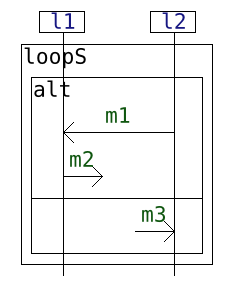}
\caption{Example interaction $i$}
\label{fig:example_int}
\end{subfigure}

\vspace*{.2cm}

\begin{subfigure}{.975\linewidth}
\captionsetup{justification=centering}
\centering
\scalebox{.7}{\begin{tikzpicture}
\node[draw,fill=white, double, double distance=2pt] (q0) {
    \begin{tikzpicture}[every node/.style = {shape=rectangle, align=center}]
\node (o) {$cr_{\{\hlf{l_2}\}}$} [level distance=0.75cm,sibling distance=1.9cm]
  child { node (o1) {$strict$} [sibling distance=.9cm]
    child { node (o11) {$\hlf{l_1}!\hms{m_1}$} }
    child { node (o12) {$\hlf{l_2}?\hms{m_1}$} }
  }
  child { node (o2) {$seq$} [sibling distance=.9cm]
    child { node (o21) {$\hlf{l_3}!\hms{m_2}$} }
    child { node (o22) {$strict$} [sibling distance=.9cm]
      child { node (o221) {$\hlf{l_3}!\hms{m_3}$} }
      child { node (o222) {$\hlf{l_2}?\hms{m_3}$} }
    }
  };
\end{tikzpicture}
};
\node[draw, circle, fill=black] (i0) at (-2.5,1.66) {};
\draw[->] (i0.east) -- (q0.west |- i0.east);
\draw[->] (q0) to [out=195,in=165,looseness=3,->]node[midway,left] {$\hlf{l_2}?\hms{m_3}$} (q0);
\node[draw,fill=white] (q2) at (4.5,-1.25) {
    \input{images/nfa/wtsimpl/term2.tex}
};
\node[draw,fill=white] (q3) at (4.5,1.25) {
    \input{images/nfa/wtsimpl/term3.tex}
};
\draw[->] (q0.east |- q2.west) -- node[midway,above] {$\hlf{l_2}!\hms{m_1}$} (q2.west);
\draw[->] (q2.north) -- node[midway,right] {$\hlf{l_1}?\hms{m_1}$} (q2.north |- q3.south);
\draw[->] (q3.west) -- node[midway,above] {$\hlf{l_1}!\hms{m_2}$} (q0.east |- q3.west);
\end{tikzpicture}}
\caption{$\nfa_s(i)$}
\label{fig:example_nfas}
\end{subfigure}
        
    \end{minipage}
    \begin{minipage}{.49\linewidth}
        \centering

        \begin{subfigure}{.975\linewidth}
        \captionsetup{justification=centering}
        \centering
        \scalebox{.7}{\begin{tikzpicture}
\node[draw, double, double distance=2pt] (q0) at (0,0) {
    \begin{tikzpicture}[every node/.style = {shape=rectangle, align=center}]
\node (o) {$cr_{\{\hlf{l_2}\}}$} [level distance=0.75cm,sibling distance=1.9cm]
  child { node (o1) {$strict$} [sibling distance=.9cm]
    child { node (o11) {$\hlf{l_1}!\hms{m_1}$} }
    child { node (o12) {$\hlf{l_2}?\hms{m_1}$} }
  }
  child { node (o2) {$seq$} [sibling distance=.9cm]
    child { node (o21) {$\hlf{l_3}!\hms{m_2}$} }
    child { node (o22) {$strict$} [sibling distance=.9cm]
      child { node (o221) {$\hlf{l_3}!\hms{m_3}$} }
      child { node (o222) {$\hlf{l_2}?\hms{m_3}$} }
    }
  };
\end{tikzpicture}
};
\node[draw, circle, fill=black] (i0)  at (-2.5,1.5) {};
\draw[->] (i0.east) -- (q0.west |- i0.east);
\node[draw, double, double distance=2pt] (q2) at (4.25,1) {
    \input{images/nfa/nosimpl/term2.tex}
};
\node[draw] (q3) at (-.5,-4) {
    \input{images/nfa/nosimpl/term3.tex}
};
\node[draw] (q4) at (3.5,-4.25) {
    \input{images/nfa/nosimpl/term4.tex}
};
\node[draw, double, double distance=2pt] (q5) at (5.5,-1.5) {
    \input{images/nfa/nosimpl/term5.tex}
};
\draw[->] (q0.east |- q2.west) -- node[midway,above] {$\hlf{l_2}?\hms{m_3}$} (q2.west);
\draw[->] (q2) to [out=345,in=15,looseness=3.75,->]node[midway,right] {$\hlf{l_2}?\hms{m_3}$} (q2);
\draw[->] (q3.north |- q0.south) -- node[midway,left] {$\hlf{l_2}!\hms{m_1}$} (q3.north);
\draw[->] (q3.east |- q4.west) -- node[midway,above] {$\hlf{l_1}?\hms{m_1}$} (q4.west);
\draw (q4.east) -- (q5.south |- q4.east);
\draw[->] (q5.south |- q4.east) -- node[midway,right] {$\hlf{l_1}!\hms{m_2}$} (q5.south);
\draw[->] (q5.120) -- node[midway,right] {$\hlf{l_2}?\hms{m_3}$} (q5.120 |- q2.south);
\draw[->] (q4) -- node[midway,right] {$\hlf{l_2}!\hms{m_1}$} (q0);
\draw[->] (q5) -- node[midway,above] {$\hlf{l_2}!\hms{m_1}$} (q0);
%
%
%
%
\end{tikzpicture}}
        \caption{$\nfa(i)$}
        \label{fig:example_nfa}
        \end{subfigure}
        
    \end{minipage}

    \caption{Example of a translation}
    \label{fig:example}

\end{figure}

Fig.\ref{fig:example} illustrates this process on an example and the advantage of using term simplification.
On Fig.\ref{fig:example_nfa}, we represent the NFA $\nfa(i)$ associated to the example interaction $i$ drawn on Fig.\ref{fig:example_int}.
We deliberately chose a very simple interaction so that the associated NFA is small and easy to visualize.
Fig.\ref{fig:example_nfas} gives the simplified $\nfa_s(i)$ obtained via the simplification of redundant occurrences of $\varnothing$.
The use of these simplifications allows a reduction of $2$ states (from $5$ in $\nfa(i)$ to $3$ in $\nfa_s(i)$).


\section{Experiments\label{sec:experiments}}

We have implemented NFA generation from interactions via the $\nfa_s$ function in the tool HIBOU \cite{hibou_label}. Simplifications (Def.\ref{def:simp_varnothing}) are systematically applied to newly encountered interactions when building the automaton (the object $\nfa(i)$ e.g. from Fig.\ref{fig:example_nfa} is never built and we directly construct $\nfa_s(i)$ e.g. from Fig.\ref{fig:example_nfas}).

The following presents three sets of experiments which main results are resp. given on Fig.\ref{fig:locks_nfa_gen}, Fig.\ref{fig:locks_nfa_compo} and Fig.\ref{fig:trace_analysis_perf_table}.
The experiments were performed using an Intel(R) Core(TM)i5-6360U CPU (2.00GHz) with 8GB RAM. 
Details of the experiments and the artifacts required to reproduce them are available in \cite{hibou_nfa_generation,hibou_nfa_trace_analysis} and in the appendices.

\textbf{Assessing the reduction of generated NFA.}
As a first step, to assess the size of a generated NFA we compare its number of states to that of equivalent minimal NFA and DFA computed using implementations of resp. the Kameda-Weiner \cite{on_the_state_minimization_of_nondeterministic_finite_automata} and Brzozowski's algorithms in the AUTOUR \cite{autour_tool} toolbox.

So as to highlight the importance of the NFA formalism w.r.t. DFA, let us consider a digital lock mechanism as a first use case.
To open the lock, we have to enter $6$ digits, the first $3$ of them forming a specific sequence $w$, while the last $3$ may be random.
Encoded as an interaction, it consists of a single lifeline $\hlf{l}$, which can receive as inputs combinations of $\hms{a}$ and $\hms{b}$ characters. Once the lock is unlocked, it emits a $\hms{u}$ message. Hence, given $L = \{\hlf{l}\}$, $M = \{\hms{a},\hms{b},\hms{u}\}$ this specification accepts traces corresponding, given alphabet $\mathbb{A}$, to the regular expression: $(\hlf{l}?\hms{a}|\hlf{l}?\hms{b})^*.w.(\hlf{l}?\hms{a}|\hlf{l}?\hms{b})^3.\hlf{l}!\hms{u}$.
This system with a single lock can be used to define systems with $n > 1$ locks in which before opening a lock, we might need to unlock first several others.
For instance, let us consider a system with $2$ locks $\hlf{l_1}$ and $\hlf{l_2}$ such that $\hlf{l_1}$ must be unlocked first.
Its behavior can be characterized using $seq(i_{\hlf{l_1}},seq(strict(\hlf{l_1}!\hms{u},\hlf{l_2}?\hms{u}),i_{\hlf{l_2}}))$ where $i_{\hlf{l_1}}$ and $i_{\hlf{l_2}}$ resp. correspond to combinations required to open $\hlf{l_1}$ and $\hlf{l_2}$.

On Fig.\ref{fig:locks_nfa_gen}, we summarize experiments performed on 3 systems with resp. 1, 4 and 8 locks (each example corresponding to a row). The $|Q|~nfa_s$, $|Q|~min_{NFA}$ and $|Q|~min_{DFA}$ columns resp. correspond to the number of states of the generated NFA and equivalent minimal NFA and DFA. The column on their right gives the time required to generate or minimize the FA. 
We can observe that determinization increases the number of states exponentially. NFA minimization yields a NFA with the same number of states for the smaller example and quickly becomes untractable for the larger ones.

\begin{figure}[h]
    \centering

\begin{subfigure}{\textwidth}
    \centering
\begin{tabular}{|l|c|c|c|c|c|c|}
\hline 
&
$|Q|$ $\nfa_s$
&
\faClock ~ $\nfa_s$
&
$|Q|$ $min_{NFA}$
&
\faClock ~ $min_{NFA}$
&
$|Q|$ $min_{DFA}$
&
\faClock ~ $min_{DFA}$
\\
\hline 
1 lock
&
$8$
&
$230~\mu$s
&
$8$
&
$30$s
&
$14$
&
$130~\mu$s
\\
\hline 
4 locks
&
$105$
&
$7400~\mu$s
&
\textcolor{black}{N/A
}
&
\textcolor{black}{timeout}
&
$312$
&
$8000~\mu$s
\\
\hline 
8 locks
&
$2881$
&
$1$s
&
\textcolor{black}{N/A
}
&
\textcolor{black}{timeout}
&
$16274$
&
$2$s
\\
\hline 
\end{tabular}

    \caption{Comparing $\nfa_s$ with equivalent minimal NFA and DFA}
    \label{fig:locks_nfa_gen}
\end{subfigure}

\vspace*{.25cm}

\begin{subfigure}{\textwidth}
    \centering
\begin{tabular}{|l|c|c|c|c|c|c|}
\hline 
&
$|Q|$ $\nfa_s$
&
\faClock ~ $\nfa_s$
&
$|Q|$ $\nfacompo$
&
\faClock ~ $\nfacompo$
&
$|Q|$ $min_{DFA}$
&
\faClock ~ $min_{DFA}$
\\
\hline 
1 lock
&
$8$
&
$230~\mu$s
&
$13$
&
$370~\mu$s
&
$14$
&
$120~\mu$s
\\
\hline 
\makecell{4 locks $s$\&$p$}
&
$97$
&
$3500~\mu$s
&
$223$
&
$7300~\mu$s
&
$298$
&
$8100~\mu$s
\\
\hline 
\makecell{8 locks $s$\&$p$}
&
$1624$
&
$0.38$s
&
$5904$
&
$0.15$s
&
$9374$
&
$2.3$s
\\
\hline 
\end{tabular}

    \caption{Comparing $\nfa_s$ with compositional NFA $\nfacompo$ generation method}
    \label{fig:locks_nfa_compo}
\end{subfigure}
    
    \caption{Usecases with networked digital locks}
    \label{fig:locks_usecases}
\end{figure}

\textbf{Comparison with compositional NFA generation.}
In order to compare our approach with those of the literature \cite{model_checking_of_message_sequence_charts,runtime_monitoring_of_web_service_conversations}, we have also implemented an alternative ``compositional'' method for NFA generation.
It involves \textbf{(1)} identifying the maximal ``basic interactions'' (i.e. with only $strict$ and $seq$) sub-terms of an initial interaction $i$, \textbf{(2)} generating a NFA for each one using our approach and \textbf{(3)} composing them using NFA operators depending on the structure of $i$ to obtain a certain $\nfacompo(i)$ NFA. 

Because there is no NFA operator for weak sequencing, a limitation of this compositional approach (which we further discuss in Sec.\ref{sec:gen_nfa}) is that $i$ cannot have $seq$ outside ``basic interactions''.
So as to apply it to our examples, we therefore need to replace the $seq$ between the locks with either $strict$ or $par$. Doing so requires identifying which locks are executed concurrently resp. strictly sequentially (which is normally done by $seq$ organically depending on the exchange of $\hms{u}$ messages between locks).
By manually replacing $seq$, we obtain a different specification which expresses a slightly different language.
We can then compare both approaches ($\nfa_s$ and $\nfacompo$) on these new $s$\&$p$ specifications.
Results are presented on Fig.\ref{fig:locks_nfa_compo}.
A NFA obtained via $\nfacompo$ has fewer states than the minimal DFA but still much more states than the NFA obtained via $\nfa_s$.

\textbf{Distributed system use cases and trace analysis.}
In this context, trace analysis simply corresponds to assessing whether or not a trace $t \in \mathbb{T}$ belongs to the semantics $\sigma(i)$ of an interaction $i$.
In previous works \cite{revisiting_semantics_of_interactions_for_trace_validity_analysis}, we used interactions directly to perform trace analysis.
Indeed, the relation $\rightarrow$ allows trying to re-enact a trace $t$ from an initial $i$, and, if this is possible, then $t \in \sigma(i)$.
As an alternative, we can leverage the generated NFA to verify $t \in \mathcal{L}(\nfa_s(i))$ which simply consists in trying to read the word in the NFA.
The major difference is that this equates to using pre-calculated interaction terms (the reachable simplified interactions) instead of computing them as needed.

\begin{figure}[h]
    \centering

\begin{tabular}{|l|c|c|c|c|c|c|c|c|}
\hline 
&
\multicolumn{2}{c}{NFA}
&
\multicolumn{2}{|c}{Traces}
&
\multicolumn{4}{|c|}{Rates in $a.s^{-1}$}
\\
\hline
&
$|Q|$
&
\faClock
&
$T$
&
$t$ sizes
&
\faClock ~ $i$ {\scriptsize \textcolor{darkspringgreen}{\faCheck}}
&
\faClock ~ $nfa_s$ {\scriptsize \textcolor{darkspringgreen}{\faCheck}}
&
\faClock ~ $i$ {\scriptsize \textcolor{red}{\faTimes}}
&
\faClock ~ $nfa_s$ {\scriptsize \textcolor{red}{\faTimes}}
\\
\hline 
ABP \cite{high_level_message_sequence_charts}
&
$68$
&
$7300\mu$s
&
$80$
&
$[250;10000]$
&
$24000$
&
$413000$
&
$20000$
&
$407000$
\\
\hline 
Platoon \cite{DupontSPARTA20}
&
$90$
&
$2900\mu$s
&
$40$
&
$[250;5000]$
&
$98000$
&
$430000$
&
$87000$
&
$325000$
\\
\hline 
HR \cite{software_engineering_for_dapp_smart_contracts_managing_workers_contracts}
&
$101$
&
$4300\mu$s
&
$40$
&
$[250;5000]$
&
$53000$
&
$440000$
&
\textcolor{black}{timeout}
&
$351000$
\\
\hline 
\end{tabular}
    
    \caption{
    Trace analysis using either interaction directly or generated NFA}
    \label{fig:trace_analysis_perf_table}

\end{figure}

Fig.\ref{fig:trace_analysis_perf_table} summarizes experiments conducted on interactions taken from the literature: 
a model of the Alternating Bit Protocol \cite{high_level_message_sequence_charts}, 
a connected platoon of autonomous Rovers \cite{DupontSPARTA20} 
and a DApp system for managing Human-Resources \cite{software_engineering_for_dapp_smart_contracts_managing_workers_contracts}.
Each use case corresponds to a row on Fig.\ref{fig:trace_analysis_perf_table}.
The second and third columns give the number of states of the generated NFA and the time required to compute it.
For each use case, we have generated a number of accepted traces and an equal number of error traces. Column $T$ gives this number while column $t$ gives the minimal and maximal size of generated traces. 
For each accepted trace (resp. error trace), both the interaction and NFA methods (respectively denoted as $i$ and $\nfa_s$) return a $Pass$ (resp. $Fail$) verdict.
The last four columns of the table correspond to rates (in actions per second) associated with the analysis of traces using either method on either accepted ({\scriptsize \textcolor{darkspringgreen}{\faCheck}}) or error ({\scriptsize \textcolor{red}{\faTimes}}) traces. As might be expected, empirical evidence goes in favor of using generated NFA for trace analysis.
For some long error traces in the HR usecase \cite{software_engineering_for_dapp_smart_contracts_managing_workers_contracts} analysis using interactions may even timeout.
Plotting the data reveals that, with the method from \cite{revisiting_semantics_of_interactions_for_trace_validity_analysis} the average rate decreases with the length of traces while this is not the case for the NFA-based method.
However, let us keep in mind that not all interactions can be translated into NFA (e.g. using weakly sequential loops from \cite{equivalence_of_denotational_and_operational_semantics_for_interaction_languages}).

\section{Comments and related work\label{sec:related_works}}

\begin{wrapfigure}{r}{5cm}
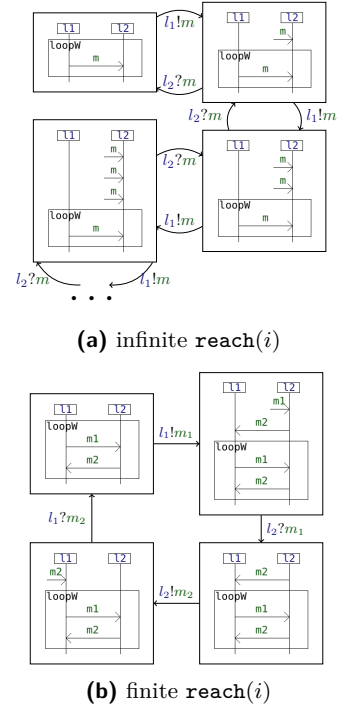

    \centering
    
    \begin{subfigure}{5cm}
        \captionsetup{justification=centering}
        \centering
        \raisebox{0pt}[\dimexpr\height-4.25\baselineskip\relax]{
        \scalebox{.6}{\input{images/loopW/unbound/graph}}
        }
        \caption{infinite $\reachableInts(i)$}
        \label{fig:loopW_reach_infite}
    \end{subfigure}

    \vspace*{.2cm}

    \begin{subfigure}{5cm}
        \captionsetup{justification=centering}
        \centering
        \scalebox{.6}{\input{images/loopW/bounded/graph}}
        \caption{finite $\reachableInts(i)$}
        \label{fig:loopW_reach_finite}
    \end{subfigure}    
        
    \caption{Boundedness \& $loop_W$}
    \label{fig:exloopW}
\end{wrapfigure}

\textbf{Repetition and boundedness.}
In addition to $loop_S$, which is related to strict sequencing $strict$, \cite{equivalence_of_denotational_and_operational_semantics_for_interaction_languages} considered $loop_W$ and $loop_P$ which are repetitions (Kleene closures) using resp. weak sequencing $seq$ and interleaving $par$.
These loops are problematic for the definition of a translation towards NFA. Indeed, in contrast to $loop_S$, with either $loop_W$ or $loop_P$, it is possible that arbitrarily many instances of the loop are active at the same time (i.e. have started but have not yet been entirely executed).
Fig.\ref{fig:loopW_reach_infite} illustrates this with a $loop_W$ which enables several instances of messages $\hms{m}$ to be sent by $\hlf{l_1}$ consecutively without any being received by $\hlf{l_2}$. 
That is $\sigma(i)$ is the set of words $w$ in $\{a= \hlf{l_1}!\hms{m}, b = \hlf{l_2}?\hms{m} \}^*$ verifying $|w|_a = |w|_b$ and for all prefixes $u$ of $w$, $|u|_a \geq |u|_b$. 
Using the pumping lemma from \cite{pumping_lemma_MadhusudanM01}, one can state that there is no NFA recognizing this language.
Another manner to consider $loop_W$ would be to restrict its use to a subclass of ``bounded interactions'' in the same fashion as ``bounded MSC graphs'' \cite{MuschollPS98,a_theory_of_regular_msc_languages,model_checking_of_message_sequence_charts,pumping_lemma_MadhusudanM01} (see also process divergence in \cite{syntaxtic_detection_of_process_divergence_and_non_local_choice_in_message_sequence_charts}) as illustrated in Fig.\ref{fig:loopW_reach_finite}. Here, because lifeline $\hlf{l_1}$ must wait for the reception of $\hms{m_2}$ before being able to emit a second instance of message $\hms{m_1}$, there are only $4$ distinct global states in which the system may be. Consequently, at most $1$ message can be in transit between $\hlf{l_1}$ and $\hlf{l_2}$. Boundedness \cite{MuschollPS98,a_theory_of_regular_msc_languages,model_checking_of_message_sequence_charts} generalizes this observation, characterizing MSC graphs in which there can only be a finite amount of ``buffered'' messages in transit between any two lifelines. Note that deciding boundedness comes at some computational cost (NP-hard, see~\cite{MuschollPS98}). 

\textbf{FA generation from interactions.}
Several works \cite{model_checking_of_message_sequence_charts,runtime_monitoring_of_web_service_conversations} propose translations of interactions into automata. In~\cite{model_checking_of_message_sequence_charts}, basic Message Sequence Charts (bMSC) are described using partial order relations induced by send-receive pairs and local orders of actions on each lifeline. A bMSC contains a sequence of send-receive pairs and corresponds, in our framework, to a sequence of patterns of the form
$strict(\hlf{l_1}!\hms{m},\hlf{l_2}?\hms{m})$ arranged with $seq$.
bMSC are then combined to form graphs of MSC (MSC-g) in order to support more complex behaviors. Nodes of a MSC-g contain bMSC and its edges, linking two bMSC correspond to a concatenation of the behaviors specified by each.
This concatenation can either be interpreted as $strict$ or $seq$ which resp. corresponds to the synchronous and asynchronous interpretation in \cite{model_checking_of_message_sequence_charts}.
For the former, MSC-g are associated to NFA by linearizing the partial orders of each bMSC node and by concatenating the resulting NFA (one per node) along the graph edges.
Besides, they point out that the problem of building a NFA from a MSC-g with the asynchronous interpretation is undecidable.
\cite{runtime_monitoring_of_web_service_conversations} proceeds in a somewhat similar fashion with UML-SD.
Basic SD (only containing message passing and $seq$) are identified within the overall SD and individually transformed into NFA. Those NFA are then composed following the structure of the SD. For this, operators of UML-SD are mapped to operators on NFA. For instance, they map $alt$ to NFA union and $loop$ to NFA Kleene star (hence it corresponds to our strictly sequential loop $loop_S$). Because there is no NFA operator equivalent to $seq$, the use of $seq$ is restricted to basic SD. This method corresponds to the $\nfacompo$ used in our experiments in Sec.\ref{sec:experiments}.

We propose translating any interaction $i$ into an equivalent NFA which states correspond to (sets of) interactions reachable from $i$ and transitions are derived from execution steps $i' \xrightarrow{a} i''$.
In our approach, every operator is handled in the same way and hence, unlike \cite{model_checking_of_message_sequence_charts,runtime_monitoring_of_web_service_conversations}, the $seq$ operator can be used above complex operators such as $alt$, $par$ or $loop_S$.

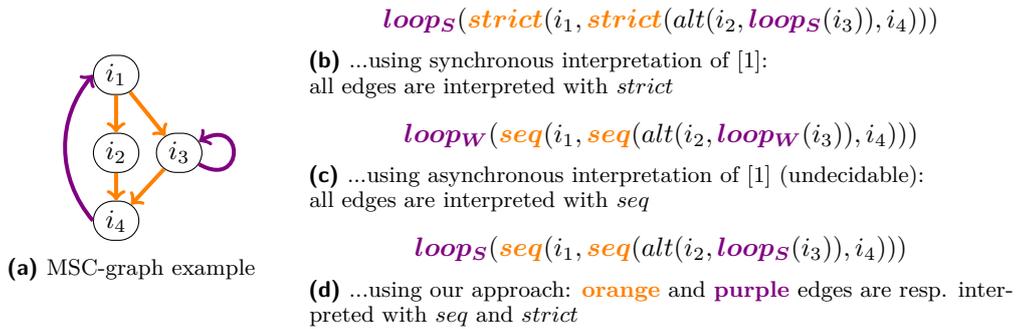
\begin{figure}[h]
    \centering

\begin{minipage}{.275\textwidth}
    \begin{subfigure}{.975\linewidth}
        \centering

\begin{tikzpicture}
\node[draw,rounded rectangle] (i1) at (0,0) {$i_1$};
\node[draw,rounded rectangle,below=0.5cm of i1] (i2) {$i_2$};
\node[draw,rounded rectangle,below right=0.5cm and .75cm of i1] (i3) {$i_3$};
\node[draw,rounded rectangle,below=1.4cm of i1] (i4) {$i_4$};
\draw (i1) edge[->,orange,line width=1.5pt] (i2);
\draw (i1) edge[->,orange,line width=1.5pt] (i3);
\draw[->,violet,line width=1.5pt] (i3) to [out=330,in=30,looseness=6,->] (i3);
\draw (i2) edge[->,orange,line width=1.5pt] (i4);
\draw (i3) edge[->,orange,line width=1.5pt] (i4);
\draw (i4.180) edge[bend left = 40,->,violet,line width=1.5pt] (i1.180);
\end{tikzpicture}
        
        \caption{MSC-graph example}
        \label{fig:msc_graph_example}
    \end{subfigure}
\end{minipage}
\begin{minipage}{.675\textwidth}

    \begin{subfigure}{.975\linewidth}
        \centering
        $\shortColViolet{\bm{loop_S}}(\shortColOrange{\bm{strict}}(i_1,\shortColOrange{\bm{strict}}(alt(i_2,\shortColViolet{\bm{loop_S}}(i_3)),i_4)))$
        \caption{...using synchronous interpretation of \cite{model_checking_of_message_sequence_charts}: \\
        all edges are interpreted with $strict$}
        \label{fig:msc_graph_equivalent_strict}
    \end{subfigure}

    \vspace*{.25cm}
    
    \begin{subfigure}{.975\linewidth}
        \centering
        $\shortColViolet{\bm{loop_W}}(\shortColOrange{\bm{seq}}(i_1,\shortColOrange{\bm{seq}}(alt(i_2,\shortColViolet{\bm{loop_W}}(i_3)),i_4)))$
        \caption{...using asynchronous interpretation of \cite{model_checking_of_message_sequence_charts} (undecidable): \\
        all edges are interpreted with $seq$}
        \label{fig:msc_graph_equivalent_weak}
    \end{subfigure}

    \vspace*{.25cm}
    
    \begin{subfigure}{.975\linewidth}
        \centering
        $\shortColViolet{\bm{loop_S}}(\shortColOrange{\bm{seq}}(i_1,\shortColOrange{\bm{seq}}(alt(i_2,\shortColViolet{\bm{loop_S}}(i_3)),i_4)))$
        \caption{...using our approach: \shortColOrange{\textbf{orange}} and \shortColViolet{\textbf{purple}} edges are resp. interpreted with $seq$ and $strict$}
        \label{fig:msc_graph_equivalent_our}
    \end{subfigure}
    
\end{minipage}
    
    \caption{MSC-graph example and equivalent interactions}
    \label{fig:msc_graph_interpretation}
\end{figure}

Fig.\ref{fig:msc_graph_example} provides an example MSC-g built over $4$ bMSC.
Under the synchronous interpretation of graph edges, we would translate this MSG-g into the interaction term from Fig.\ref{fig:msc_graph_equivalent_strict}.
If we were to consider the asynchronous interpretation, using the weak sequential loop operator $loop_W$ from \cite{equivalence_of_denotational_and_operational_semantics_for_interaction_languages}, we would translate it as the term from Fig.\ref{fig:msc_graph_equivalent_weak}.
Unlike the former, the latter cannot be translated into a NFA, as stated by the non-decidability result from \cite{model_checking_of_message_sequence_charts}.
However, with our approach, we can provide NFA translation with a greater expressivity, as underlined with the third interpretation of the graph given on Fig.\ref{fig:msc_graph_equivalent_our}.
Indeed, using $loop_S$ and $seq$, we may interpret distinct edges of the graph either synchronously or asynchronously.
Moreover, this allows using $seq$ at any scope within the specification, unlike \cite{runtime_monitoring_of_web_service_conversations} where it must be used within basic SD.

The principle of our approach lies much closer to that of \cite{partial_derivatives_of_regular_expressions_and_finite_automata_contructions,generating_optimal_monitors_for_extended_regular_expressions} in which FA are build from regular expressions.
\cite{generating_optimal_monitors_for_extended_regular_expressions} extends upon \cite{testing_extended_regular_language_membership_incrementally_by_rewriting} in which, for any regular expression and any letter, a unique derivative expression is defined (in effect constituting a deterministic operational semantics of regular expressions).
\cite{generating_optimal_monitors_for_extended_regular_expressions} uses circular coinduction as a decision procedure for the equivalence of regular expressions. It then leverages the unicity of derivatives in \cite{testing_extended_regular_language_membership_incrementally_by_rewriting} to construct a DFA which states correspond to circularities of regular expressions (i.e. groups of regular expressions with the same language) and transitions to the application of derivations.
By contrast, in \cite{partial_derivatives_of_regular_expressions_and_finite_automata_contructions}, partial derivatives are used on regular expressions written in linear form so as to obtain small NFA (with a number of states no greater than the alphabetic width of the regular expression).

The size of NFA generated from basic interactions (with only $seq$ and message passing via $strict$) is, in the worst case, $O(n^k)$ where $n$ is the number of message passing patterns and $k$ is the number of lifelines~\cite{model_checking_of_message_sequence_charts,runtime_monitoring_of_web_service_conversations,complexity_interleaving_words_MAYER1994,complexity_regx_interleaving_Gelade10}. Such interactions resembles basic MSC~\cite{model_checking_of_message_sequence_charts,runtime_monitoring_of_web_service_conversations} or regular expressions with interleaving over finite words~\cite{complexity_interleaving_words_MAYER1994}. 
The witness interaction $par(\hlf{l_1}!\hms{m},par(\ldots,\hlf{l_n}!\hms{m})\ldots)$ consisting of the interleaving of $n$ distinct atomic actions accepts as traces all possible interleavings of those $n$ actions. Any NFA associated to it must have at least $2^n$ states.
Hence, the cost of NFA generation with $par$ is $\Omega(2^n)$, which relates to results established for regular expressions with interleaving \cite{complexity_interleaving_words_MAYER1994,complexity_regx_interleaving_Gelade10}.

\textbf{FA and trace analysis.}
FA are extensively used for trace analysis in the domain of monitoring or runtime verification. However, they are mostly generated from logics such as Linear Temporal Logic (LTL)~\cite{efficient_monitoring_of_omega_langauge,lossy_trace_JoshiTF17,lossy_trace_KauffmanHF19} rather than interaction models. \cite{efficient_monitoring_of_omega_langauge} focuses on generating FA monitors that recognize the minimal bad prefixes of $\omega$-languages (sets of infinite traces). These monitors are created by reducing Büchi Automata into Binary Transition Tree Finite State Machines (BTT-FSM). Finite trace analysis against the generated FSM is performed online (as can we using our generated NFA). 
Overall, automata are preferred in runtime verification due to their compact form, efficient memory usage, and faster processing. They offer various advantages for trace analysis, particularly for distributed systems dealing with unreliable communication channels~\cite{lossy_trace_KauffmanHF19} or lossy traces due to network-related issues~\cite{lossy_trace_JoshiTF17}.

\section{Conclusion\label{sec:conclusion}}

We have proposed an approach for generating Non-deterministic Finite Automata (NFA) from Interactions.
It employs 
derivatives and term simplification to generate concise NFA on-the-fly. 
Our approach is validated through proofs and experiments with several interactions from the literature, showcasing the reduced quality of generated NFA (no need for costly NFA reduction machinery) and their suitability for trace analysis.

\bibliography{biblio/biblio,biblio/biblio_fa_minimization,biblio/biblio_int_to_fa,biblio/biblio_monitoring}

\clearpage 

\appendix

\section{NFA generation experiments on digital locks usecases\label{anx:nfagen}}

\subsection{Comparison to minimal DFA}

A known example for showing that an equivalent minimal DFA may have exponentially more state that the associated minimal NFA is the following:
given the alphabet $\{a,b\}$, the regular expression $(a|b)^*.a.(a|b)^n$ can be encoded as a NFA with $n+2$ states while the associated minimal DFA has $2^{n+1}$ states. This is due to the fact that this construction must keep track of the last $n$ states.

\begin{figure}[h]
    \centering

\begin{tabular}{|c|c|c|}
\hline
\Large\textbf{$i$}
&
\Large\textbf{$nfa_s(i)$}
&
\Large\textbf{$min_{DFA}(nfa_s(i))$}
\\
\hline 
\includegraphics[scale=.3]{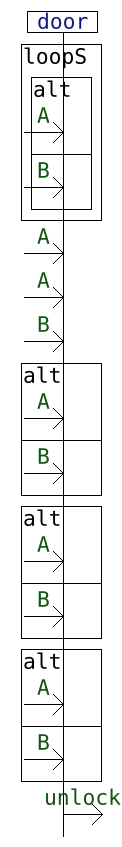}
&
\includegraphics[scale=.35]{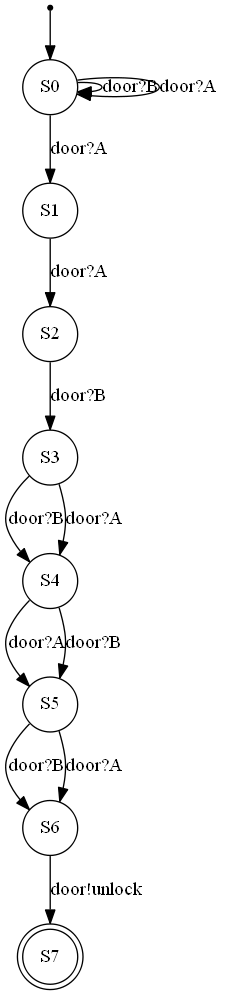}
&
\includegraphics[scale=.275]{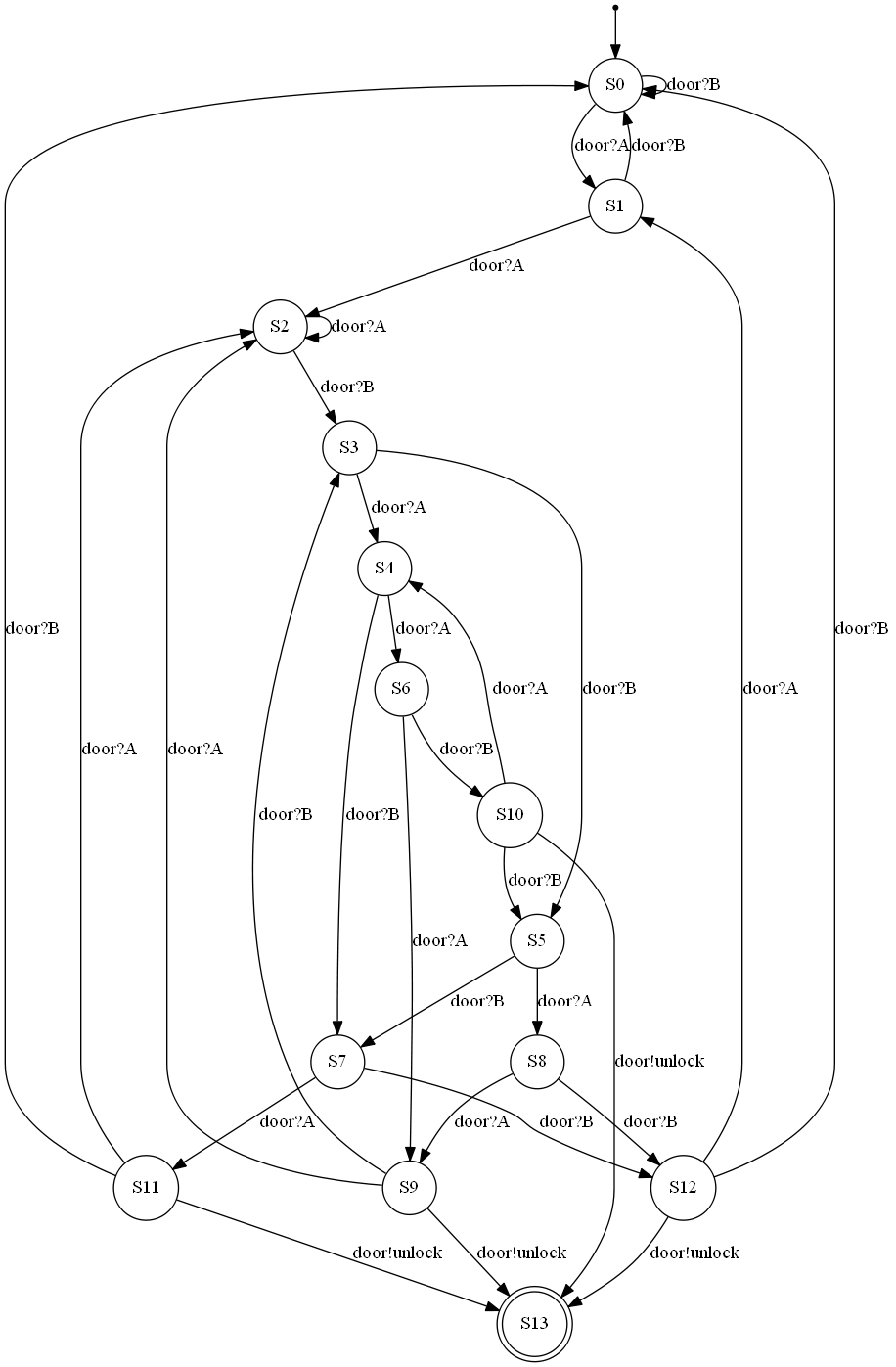}
\\
\hline 
\end{tabular}
    
    \caption{Example with a single digital lock}
    \label{fig:1lock}
\end{figure}

Taking inspiration from this let us consider the following toy usecase. We consider a door which is initially closed. It can be opened by entering a combination of $A$ and $B$ of the form $A.A.B.\#.\#.\#$ i.e. the code is $AAB$ followed by any three characters. In order to open the door we may therefore enter any word of the regular language $(A|B)^*.A.A.B.(A|B)^3$.

In order to encode this specification as an interaction, we consider the door to be a lifeline called $\hlf{\mathtt{door}}$ which can receive $\hms{\mathtt{A}}$ and $\hms{\mathtt{B}}$ messages as inputs and, once it is unlocked, it emits message $\hms{\mathtt{unlock}}$. This interaction is drawn on the left column of Fig.\ref{fig:1lock}.

Using interaction execution with term simplification to explore its semantic with memoization of previously encountered terms we obtain the NFA $\nfa_s(i)$ drawn in the middle column of Fig.\ref{fig:1lock}.
This NFA turns out to be the minimal NFA which can express this language. As expected, if we determinize it and then minimize the resulting DFA we get a FA with more states (here $14$ states instead of $8$). This DFA is given on the right column of Fig.\ref{fig:1lock}.

\begin{figure}[h]
    \centering

\begin{tabular}{|c|c|}
\hline 
Network structure & generated NFA\\
\hline 
\makecell{
\includegraphics[width=.3\textwidth]{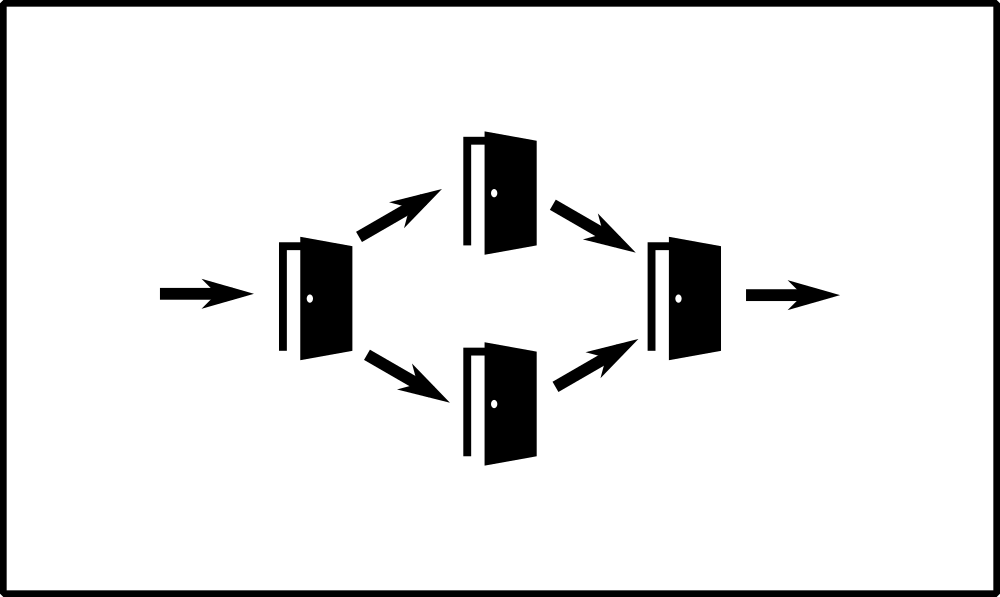}
}
&
\makecell{\includegraphics[width=.6\textwidth]{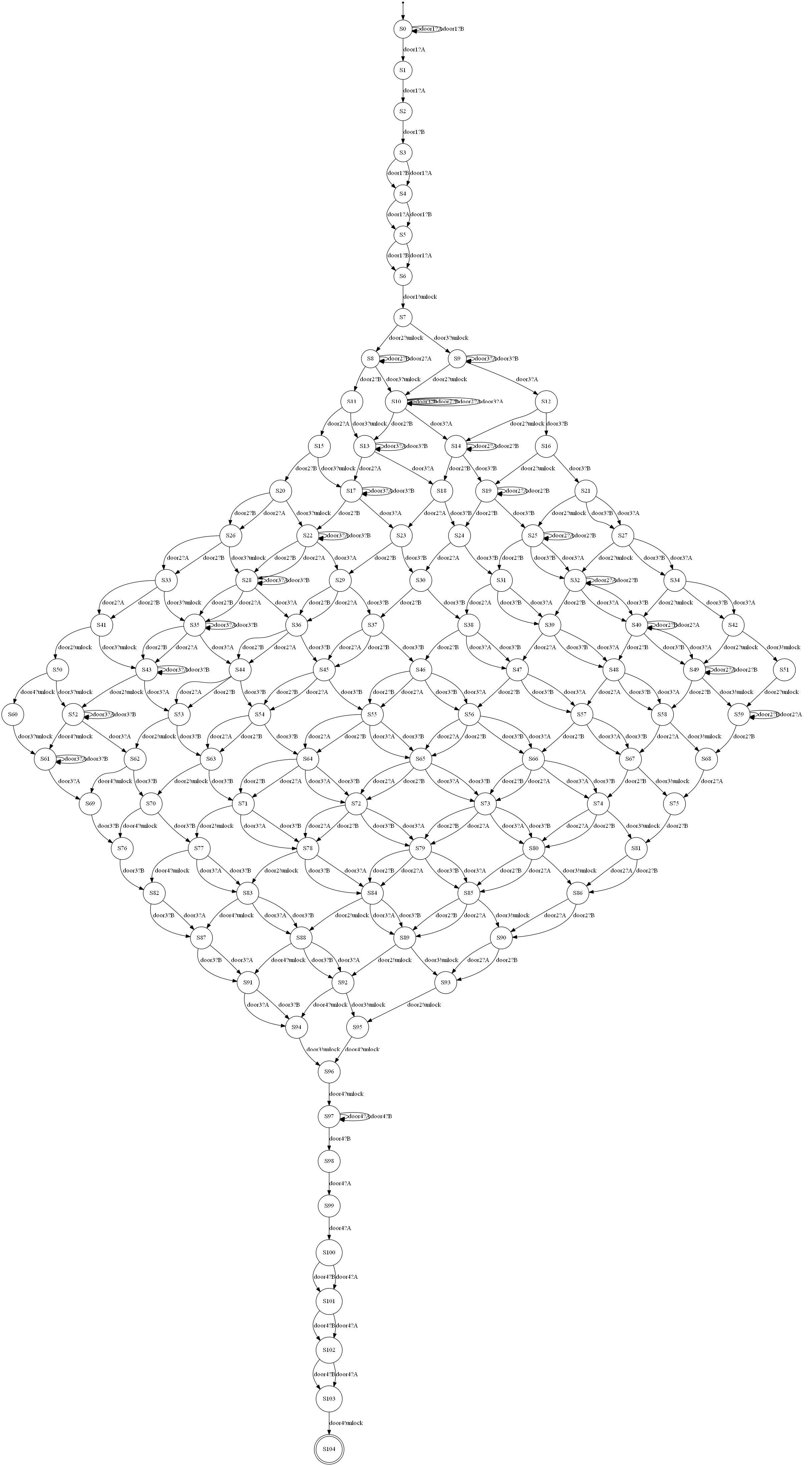}}
\\
\hline 
\end{tabular}

    \caption{NFA generated from a network with $4$ digital locks}
    \label{fig:4locks}
\end{figure}

As a distributed/concurrent system, this usecase can be described as a single door having inputs and outputs. More complex structure of linked doors can be formed, with the unlocking of a door allowing another one to be opened and so on.
The left of Fig.\ref{fig:4locks} describes one such system with $4$ locks.
Opening the first one allows trying to enter the codes for the 2 following doors. Then, we need to open these two doors to start trying to open the final door. 
On the right of Fig.\ref{fig:4locks} the generated NFA corresponding to this specification is given.

\begin{figure}[h]
    \centering

\begin{tabular}{|c|c|}
\hline 
Network structure & generated NFA\\
\hline 
\makecell{\includegraphics[width=.3\textwidth]{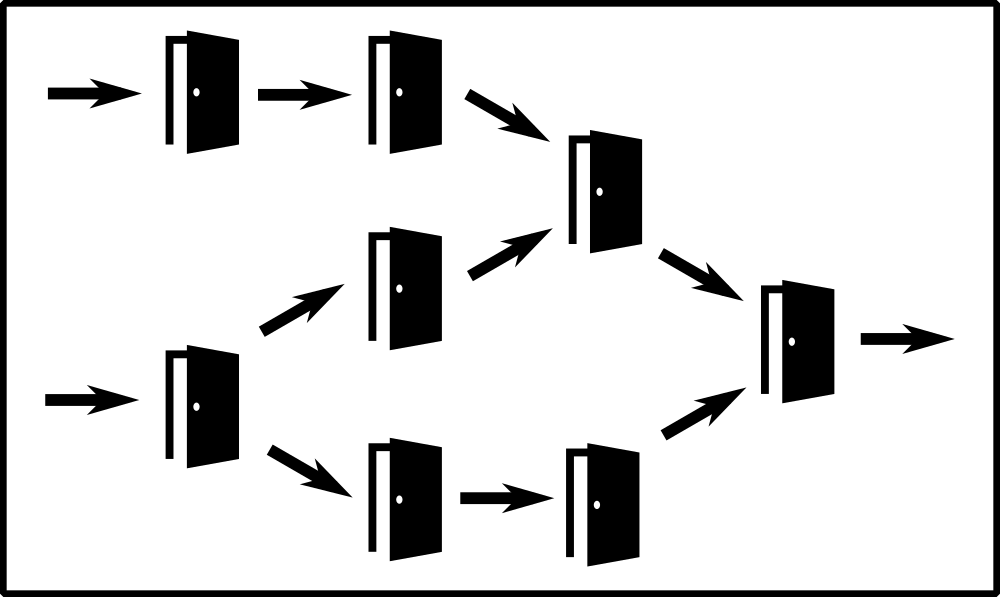}}
&
\makecell{
\Huge not drawn\\
\Huge$2881$ states}
\\
\hline 
\end{tabular}

    \caption{NFA generated from a network with $8$ digital locks}
    \label{fig:8locks}
\end{figure}

Fig.\ref{fig:8locks} illustrates our second network example, which involves $8$ locks. The generated NFA, which is not drawn for a lack of space, has $2881$ states.

\clearpage

\subsection{Comparison with compositional approach}

\begin{figure}[h]
    \centering

\begin{tabular}{|c|c|}
\hline 
\makecell{$nfa_s(i)$}
&
\makecell{$compo(i)$}
\\
\hline 
\makecell{\includegraphics[scale=.275]{appendix/nfa_gen/secret_code_orig_nfa.png}}
&
\makecell{\includegraphics[scale=.275]{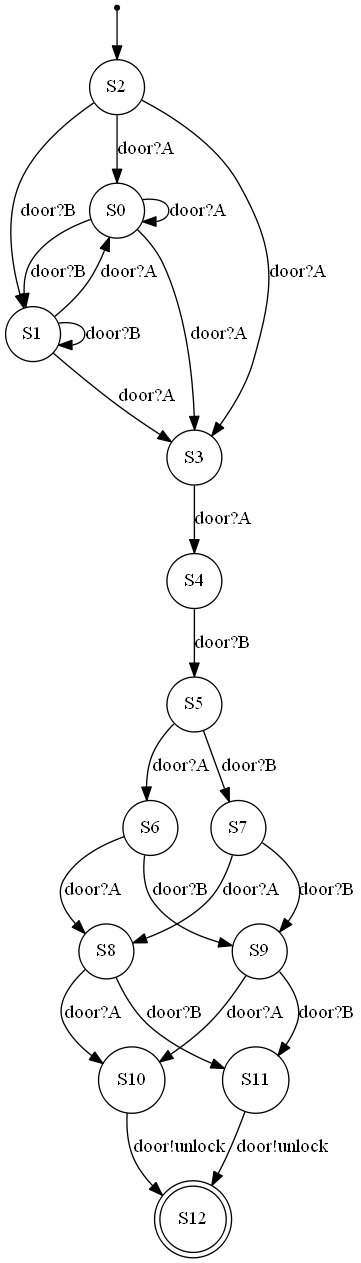}}
\\
\hline 
\end{tabular}
    
    \caption{Comparison of generated NFA for the example with 1 lock using our method and a compositional method}
    \label{fig:1lock_compo}
\end{figure}

We also compare the NFA generated using our method with NFA obtained using a compositional method. Fig.\ref{fig:1lock_compo} represents the NFA generated using both methods for our example with $1$ lock. While $\nfa_s(i)$ has $8$ states, $\nfacompo(i)$ has $13$ states.
We can remark that $\nfacompo(i)$ has less states than an equivalent minimal DFA ($14$ states as seen on Fig.\ref{fig:1lock}) but still more than $\nfa_s(i)$ which has the same number of states ($8$) as the equivalent minimal NFA.

\begin{figure}[h]
    \centering

\begin{tabular}{|c|c|c|}
\hline 
\makecell{$nfa_s(i)$}
&
\makecell{$compo(i)$}
\\
\hline 
\makecell{\includegraphics[width=.3\textwidth]{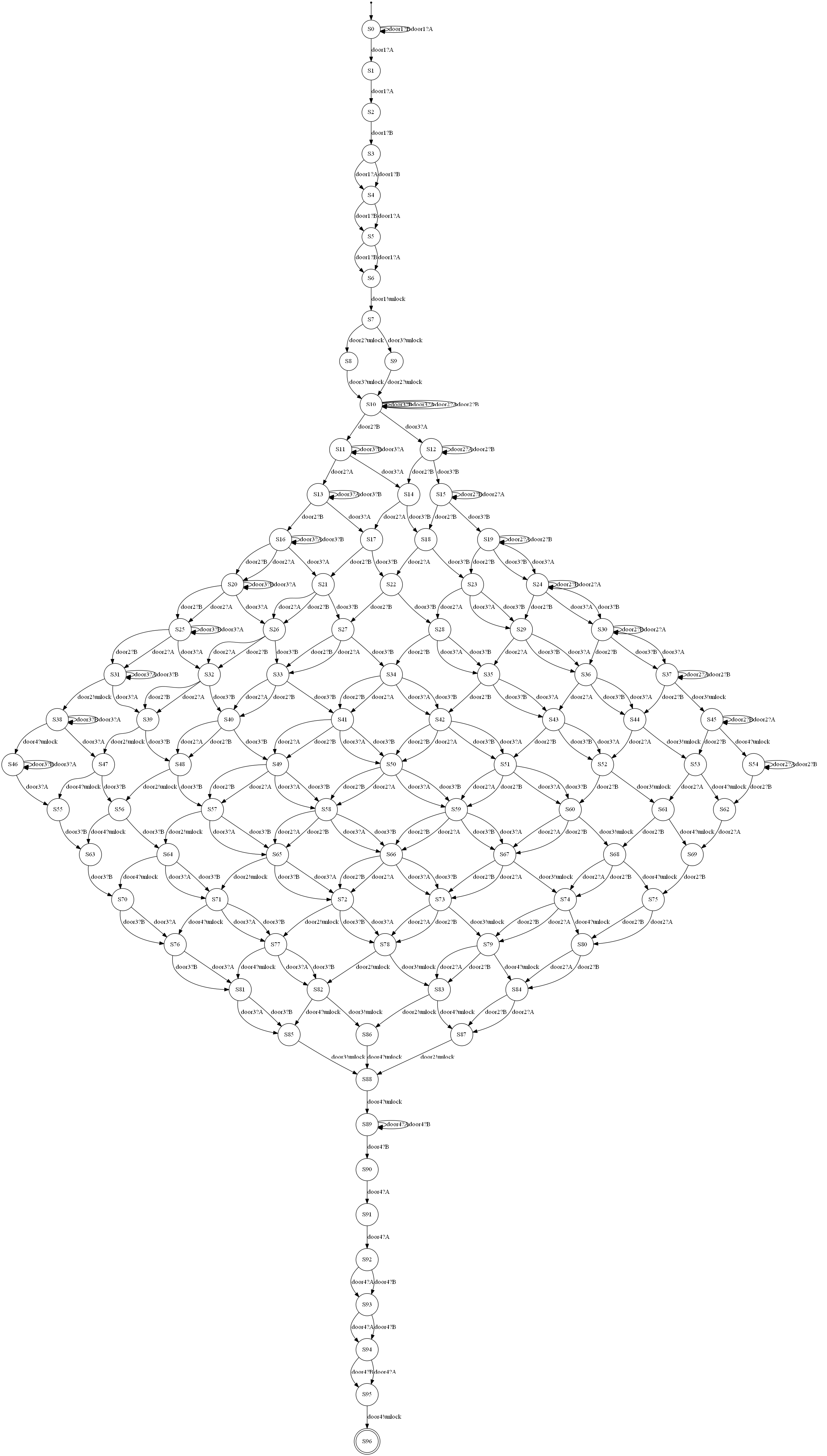}}
&
\makecell{\includegraphics[width=.6\textwidth]{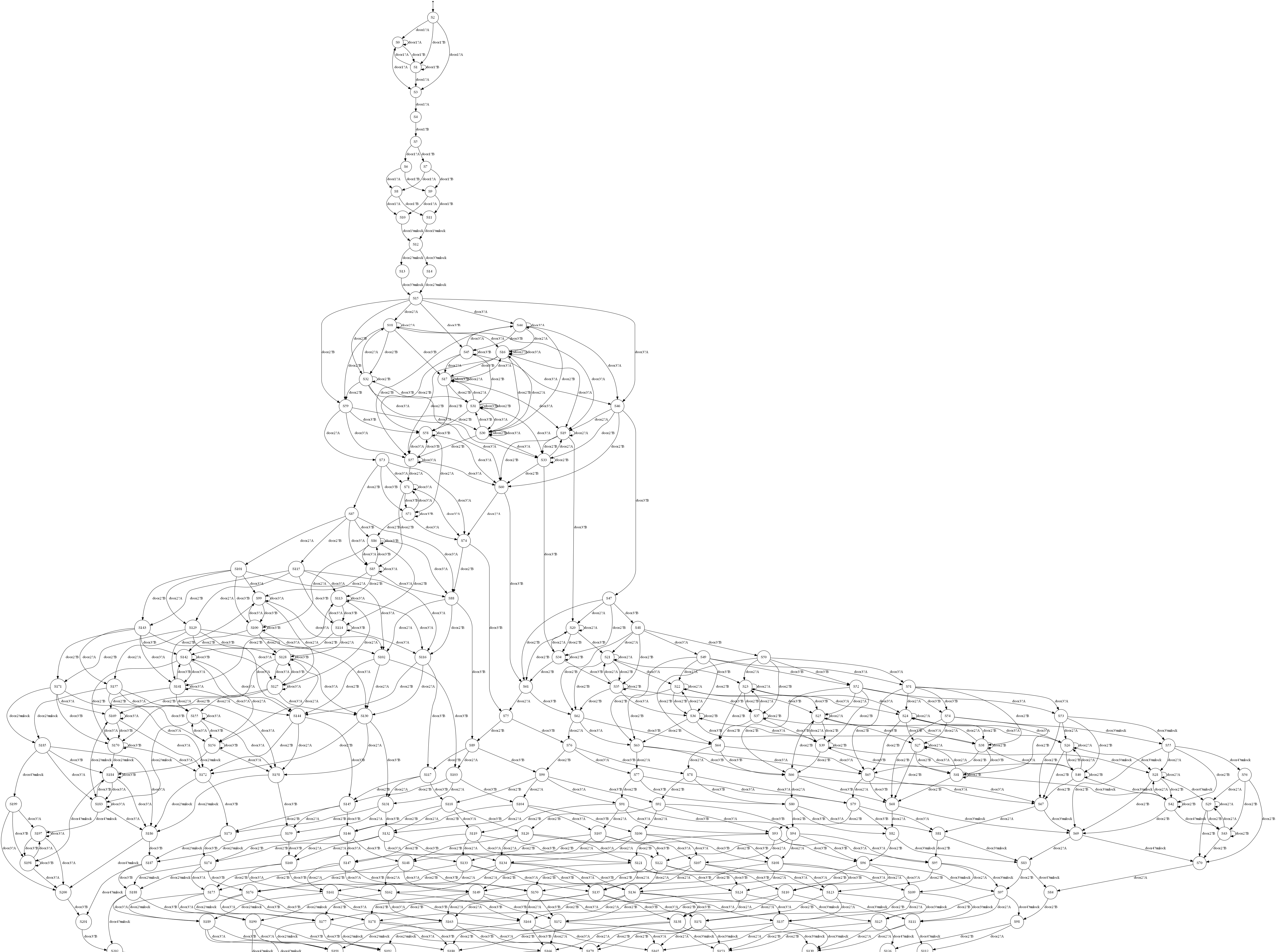}}
\\
\hline 
\end{tabular}
    
    \caption{Comparison of generated NFA for the example with $4$ locks and substituting $seq$ with either $strict$ or $par$ using our method and a compositional method.}
    \label{fig:strict_and_par_4locks_nfagen}
\end{figure}

Fig.\ref{fig:strict_and_par_4locks_nfagen} likewise represents $\nfa_s(i)$ and $\nfacompo(i)$ for the example network with $4$ locks. While the former has $97$ states, the latter has $223$ states.

\clearpage

\section{Trace analysis experiments}

In the following we consider several usecases. Each usecase is modelled by an interaction which specifies its set of accepted behaviors.

Using the trace generation feature of HIBOU \cite{hibou_label} we generate three datasets of traces, one for each example usecase. 

Each dataset consists of:
\begin{itemize}
    \item a set of accepted prefixes
    \item a set of error traces obtained by adding random actions at the end of accepted traces
\end{itemize}

We then analyze these traces against the specifying interaction using two methods:
\begin{itemize}
    \item the trace analysis algorithm from \cite{revisiting_semantics_of_interactions_for_trace_validity_analysis} which directly analyze traces against interactions via the execution relation $\rightarrow$
    \item a simple NFA word analysis technique applied to the NFA generated using $\nfa_s$
\end{itemize}

We compare the results of analyzing the traces using both methods and make sure that:
\begin{itemize}
    \item for the accepted prefixes, both methods return a Pass
    \item for the error traces, both methods return a Fail
\end{itemize}

\subsection{The Platoon \cite{DupontSPARTA20} usecase}

\begin{figure}[h]
    \centering

\begin{tabular}{|c|c|}
\hline
\Large\textbf{$i$}
&
\Large\textbf{$nfa_s(i)$}
\\
\hline 
\makecell{
\includegraphics[width=.4\textwidth]{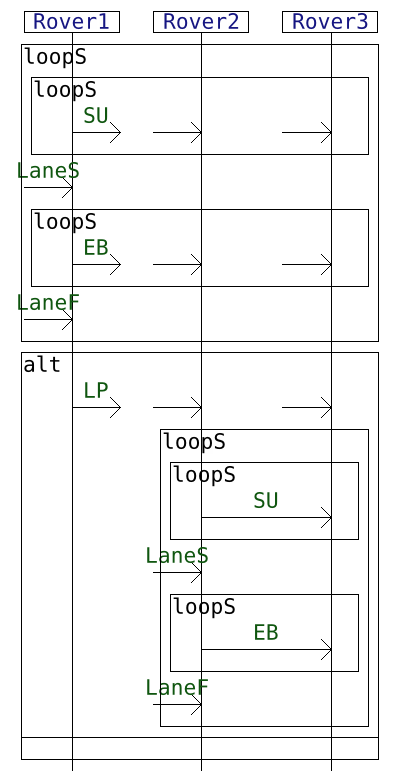}
}
&
\makecell{
\includegraphics[width=.575\textwidth]{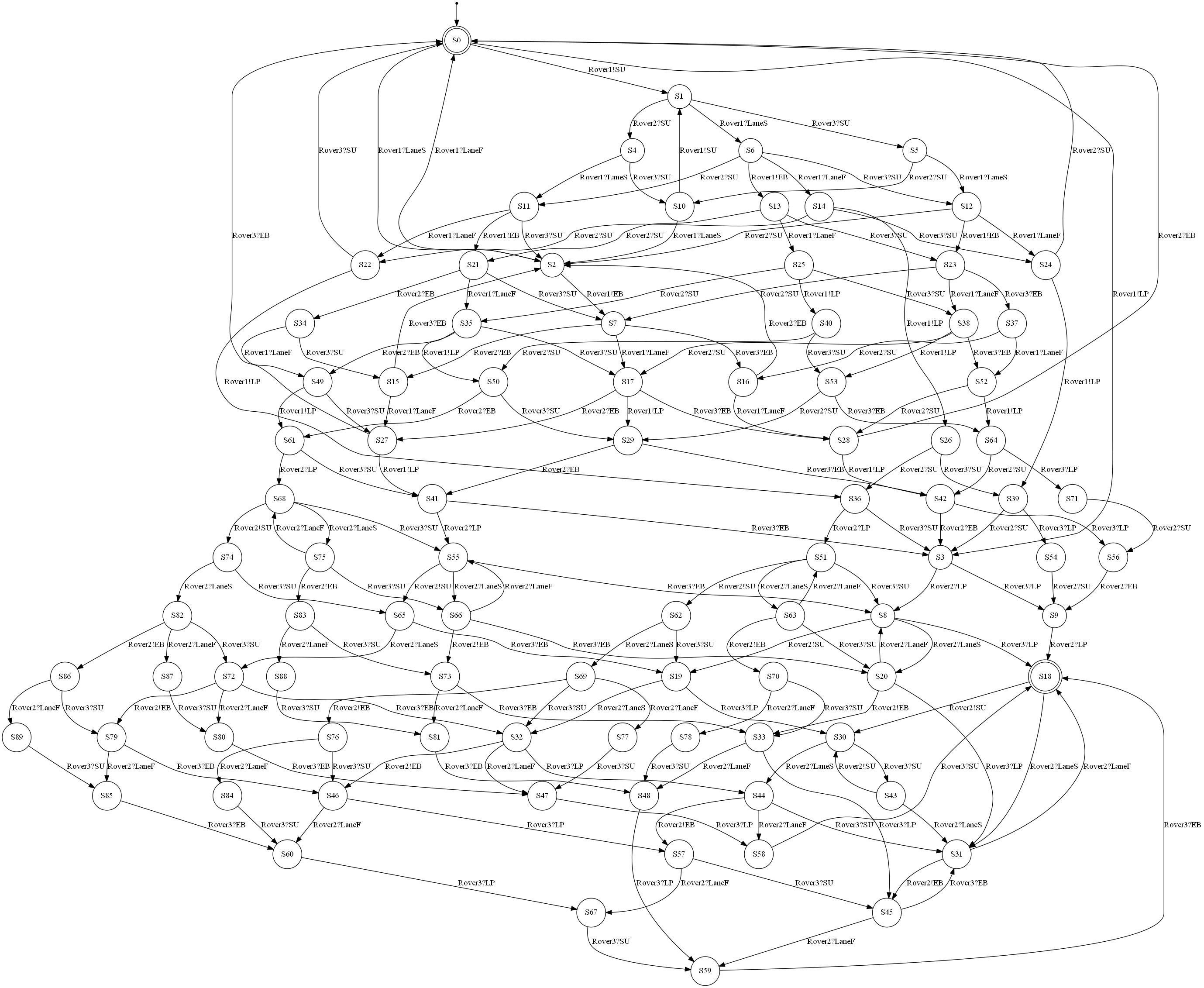}
}
\\
\hline 
\end{tabular}
    
    \caption{Platoon \cite{DupontSPARTA20} usecase}
    \label{fig:rover_usecase_desc}
\end{figure}

The interaction on the left of Fig.\ref{fig:rover_usecase_desc} describes the Platoon \cite{DupontSPARTA20} usecase.
It consists of three self-driving rovers that drive in formation.
The NFA generated from this interaction is drawn on the right of Fig.\ref{fig:rover_usecase_desc}.

Fig.\ref{fig:rover_acp} plots the results of trace analysis for the accepted prefixes and Fig.\ref{fig:rover_err} for the error traces. In these plots, each point corresponds to a given trace. Its position corresponds to the time taken to analyze it (on the $y$ axis), and its length i.e. total number of actions (on the $x$ axis, with some jitter added to better see distinct points). The time on the $y$ axis is given on a logarithmic scale.
The color represents the method used, blue for interaction global trace analysis and green for nfa word analysis.

\begin{figure}[h]
    \centering
    \includegraphics[width=.9\textwidth]{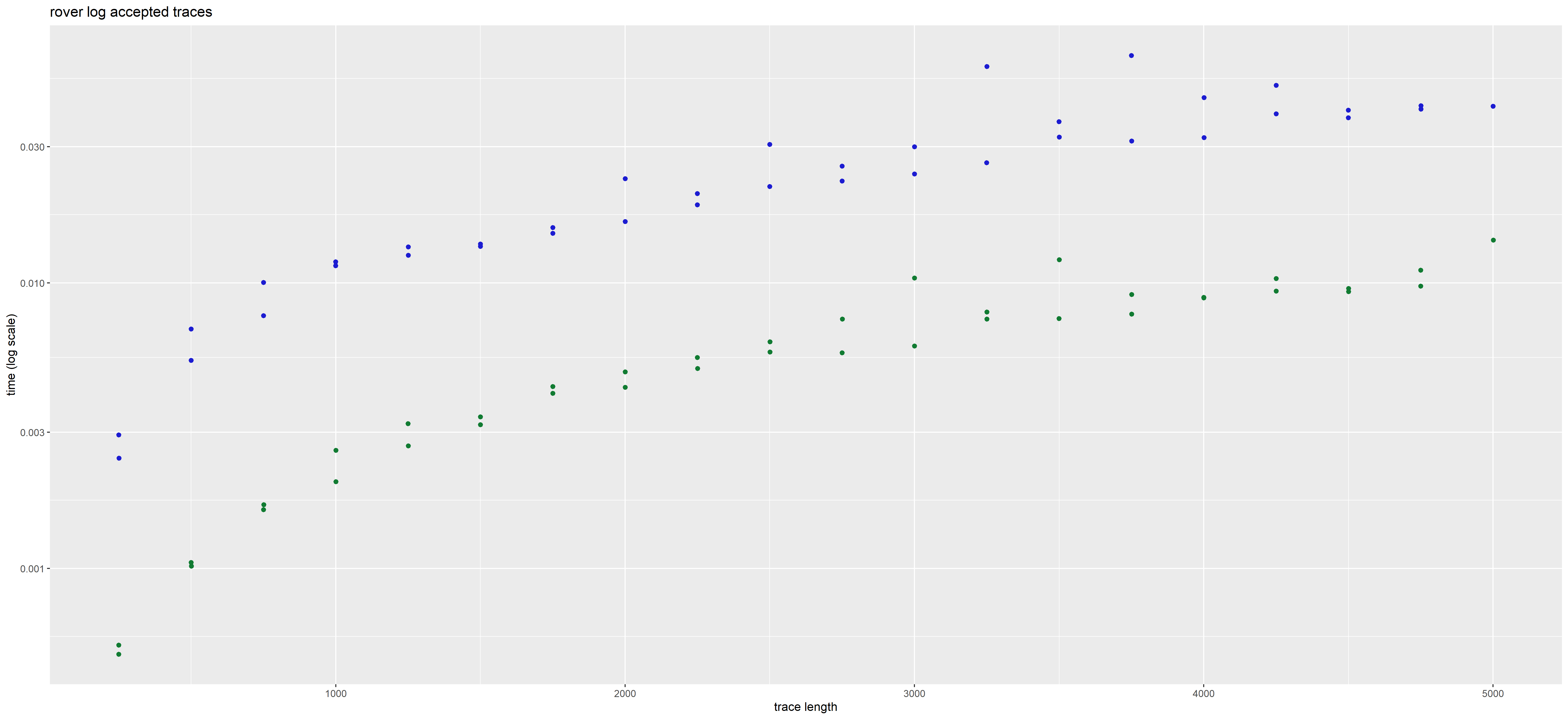}
    \caption{Accepted traces for the Platoon \cite{DupontSPARTA20} usecase}
    \label{fig:rover_acp}
\end{figure}

\begin{figure}[h]
    \centering
    \includegraphics[width=.9\textwidth]{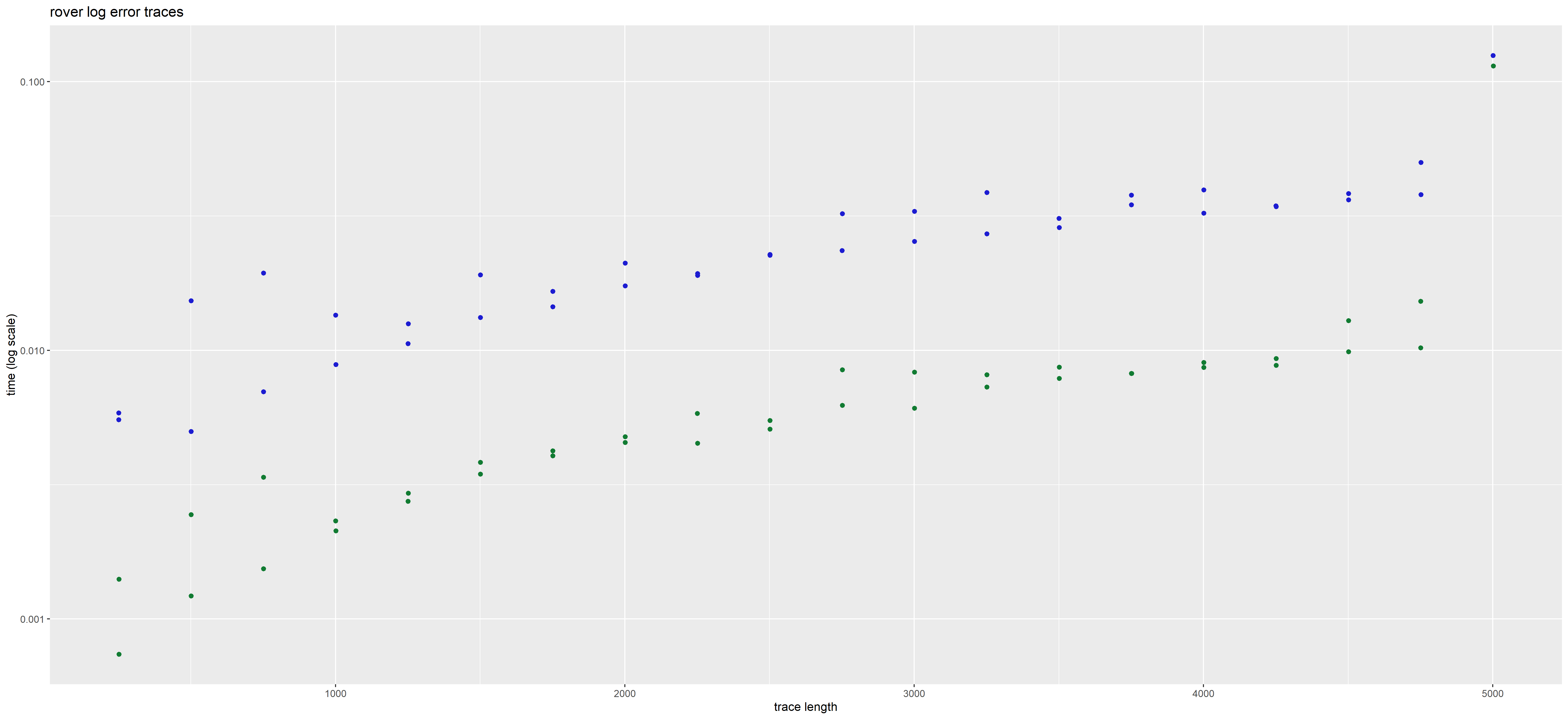}
    \caption{Error traces for the Platoon \cite{DupontSPARTA20} usecase}
    \label{fig:rover_err}
\end{figure}

\clearpage

\subsection{Human Resources \cite{software_engineering_for_dapp_smart_contracts_managing_workers_contracts} usecase}

\begin{figure}[h]
    \centering

\begin{tabular}{|c|c|}
\hline
\Large\textbf{$i$}
&
\Large\textbf{$nfa_s(i)$}
\\
\hline 
\makecell{
\includegraphics[width=.4\textwidth]{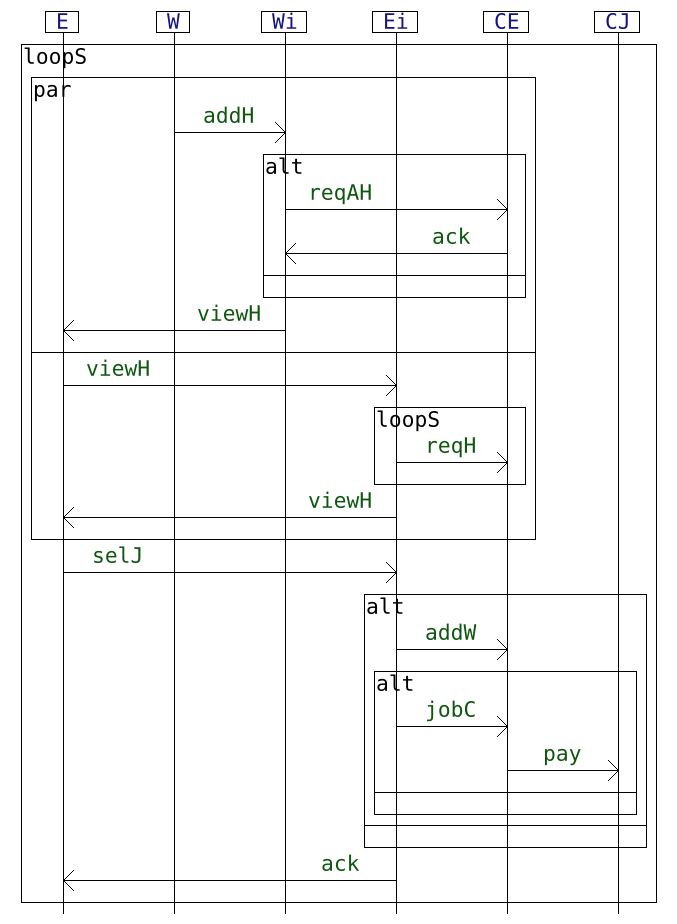}
}
&
\makecell{
\includegraphics[width=.575\textwidth]{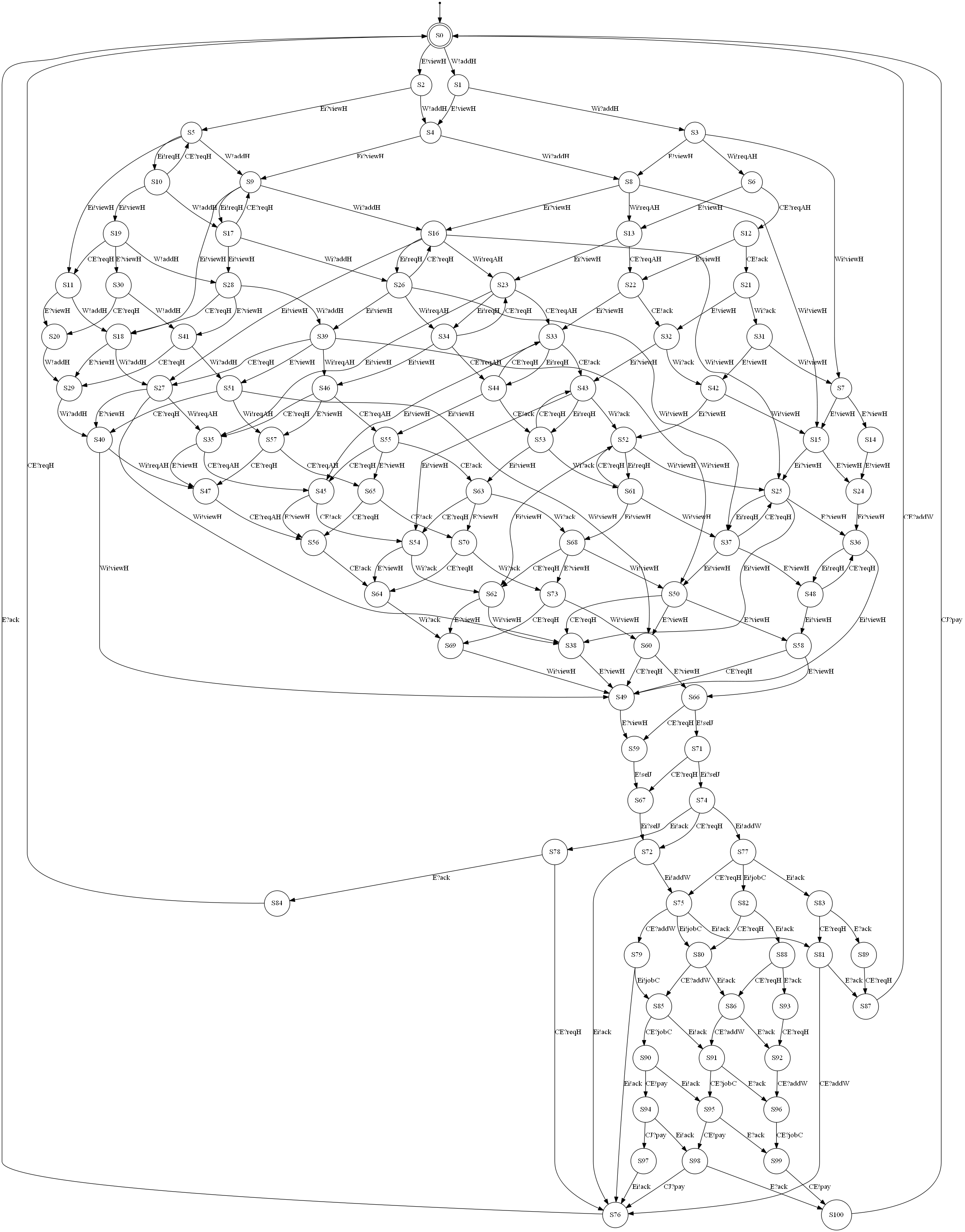}
}
\\
\hline 
\end{tabular}
    
    \caption{Human Resources \cite{software_engineering_for_dapp_smart_contracts_managing_workers_contracts} usecase}
    \label{fig:hr_usecase_desc}
\end{figure}

Similarly, while Fig.\ref{fig:hr_usecase_desc} describes the expected behaviors of the Human Resources \cite{software_engineering_for_dapp_smart_contracts_managing_workers_contracts} usecase, experimental results for trace analysis are given on Fig.\ref{fig:hr_acp} (for accepted traces) and Fig.\ref{fig:hr_err} (for error traces).

\begin{figure}[h]
    \centering
    \includegraphics[width=.9\textwidth]{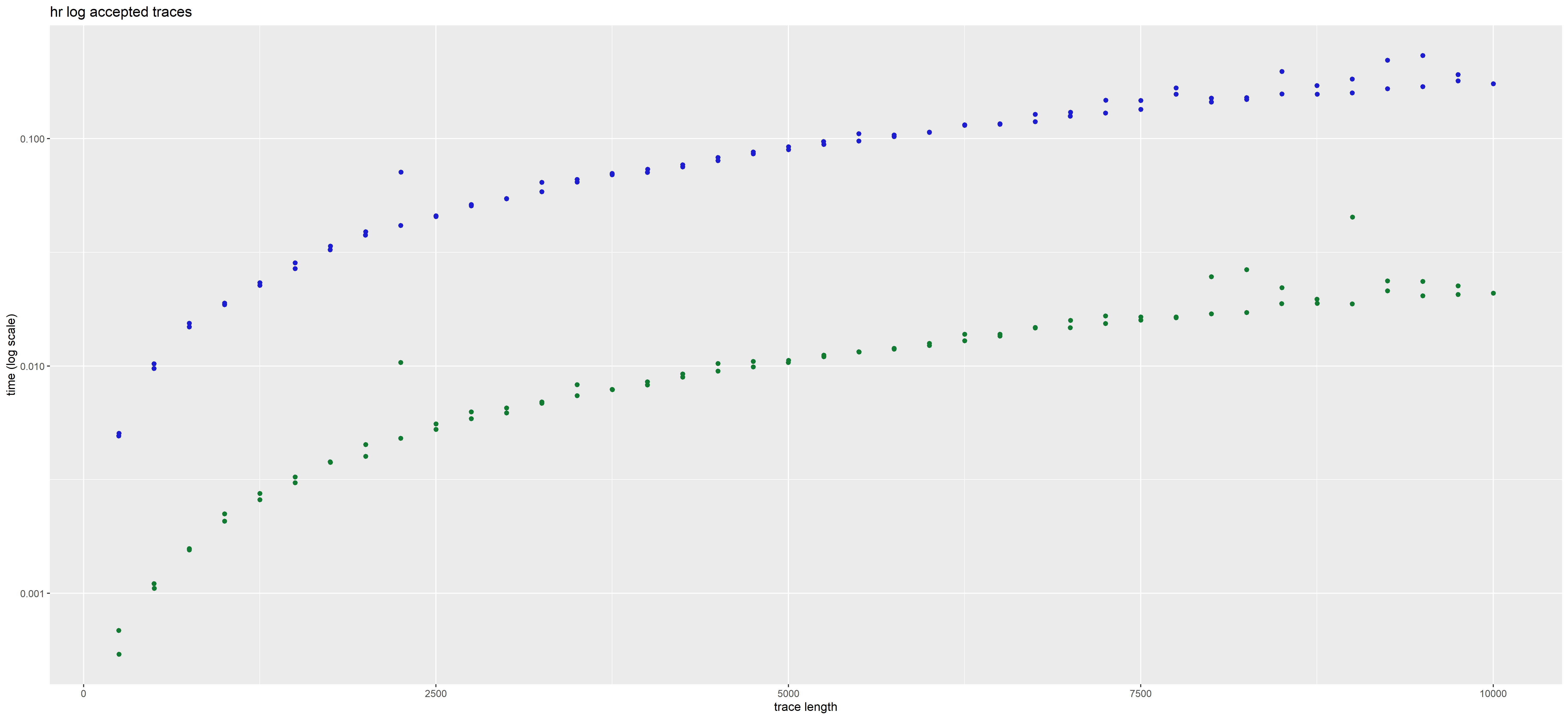}
    \caption{Accepted traces for the Human Resources \cite{software_engineering_for_dapp_smart_contracts_managing_workers_contracts} usecase}
    \label{fig:hr_acp}
\end{figure}

\begin{figure}[h]
    \centering
    \includegraphics[width=.9\textwidth]{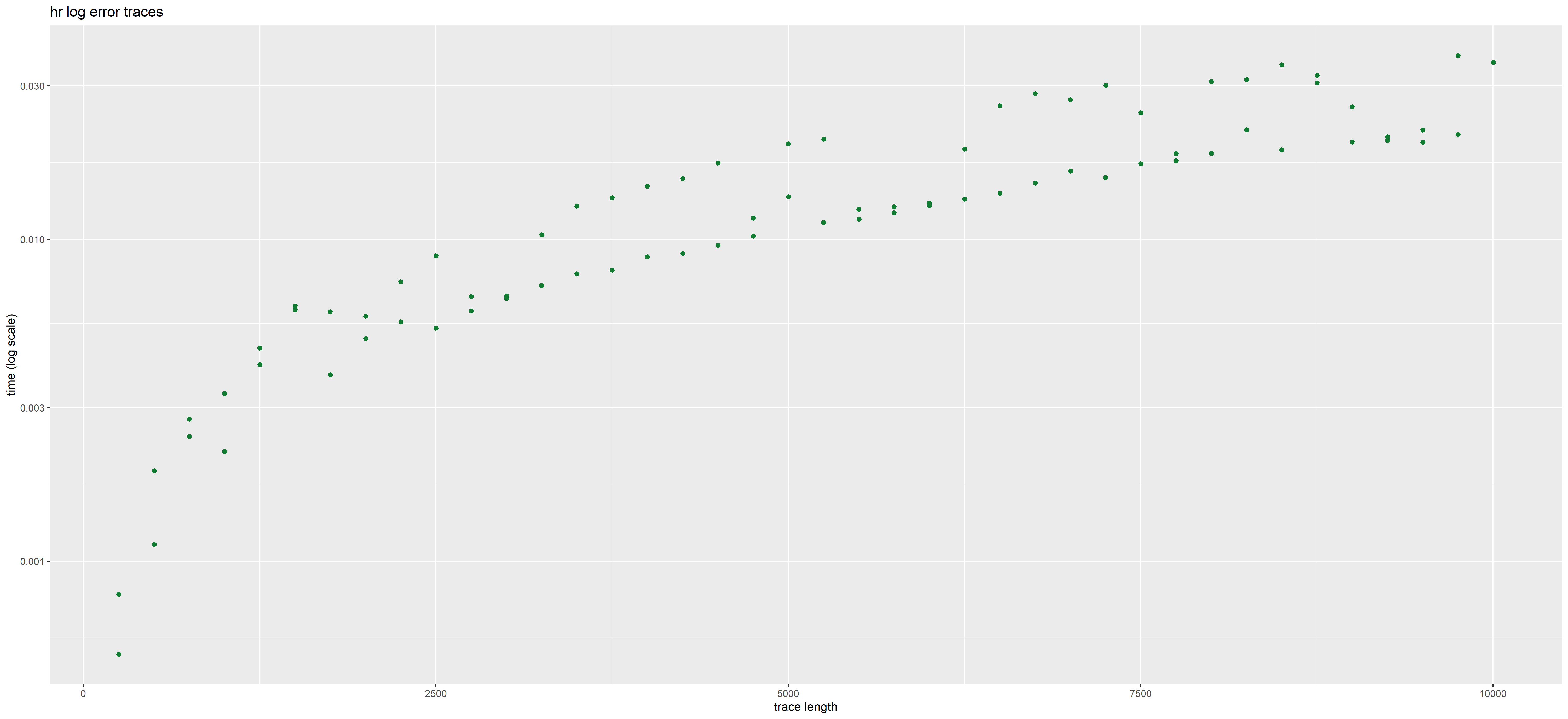}
    \caption{Error traces for the Human Resources \cite{software_engineering_for_dapp_smart_contracts_managing_workers_contracts} usecase}
    \label{fig:hr_err}
\end{figure}

\clearpage

\subsection{Alternating Bit Protocol \cite{high_level_message_sequence_charts} usecase}

\begin{figure}[h]
    \centering

\begin{tabular}{|c|c|}
\hline
\Large\textbf{$i$}
&
\Large\textbf{$nfa_s(i)$}
\\
\hline 
\makecell{
\includegraphics[width=.25\textwidth]{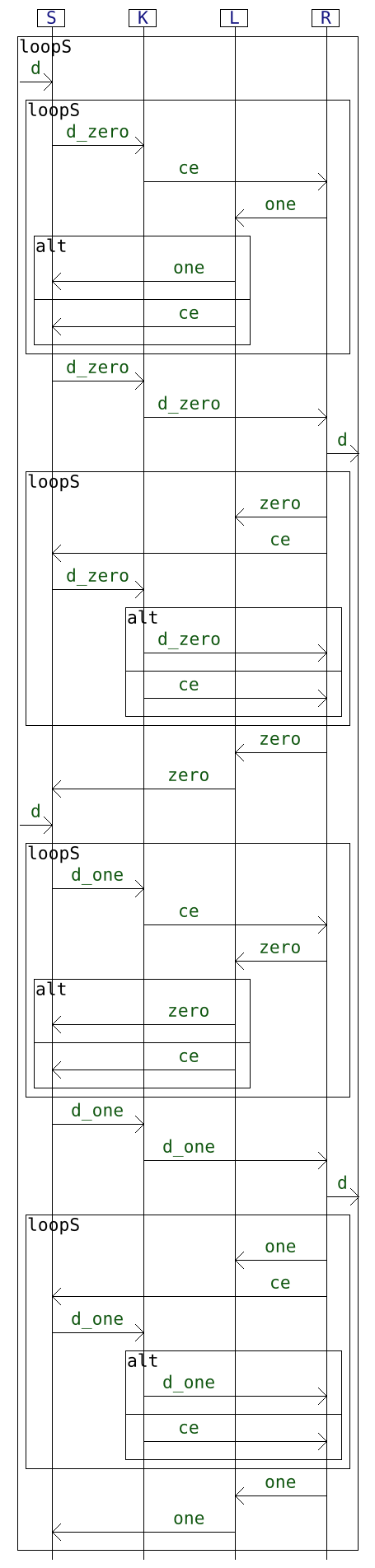}
}
&
\makecell{
\includegraphics[width=.25\textwidth]{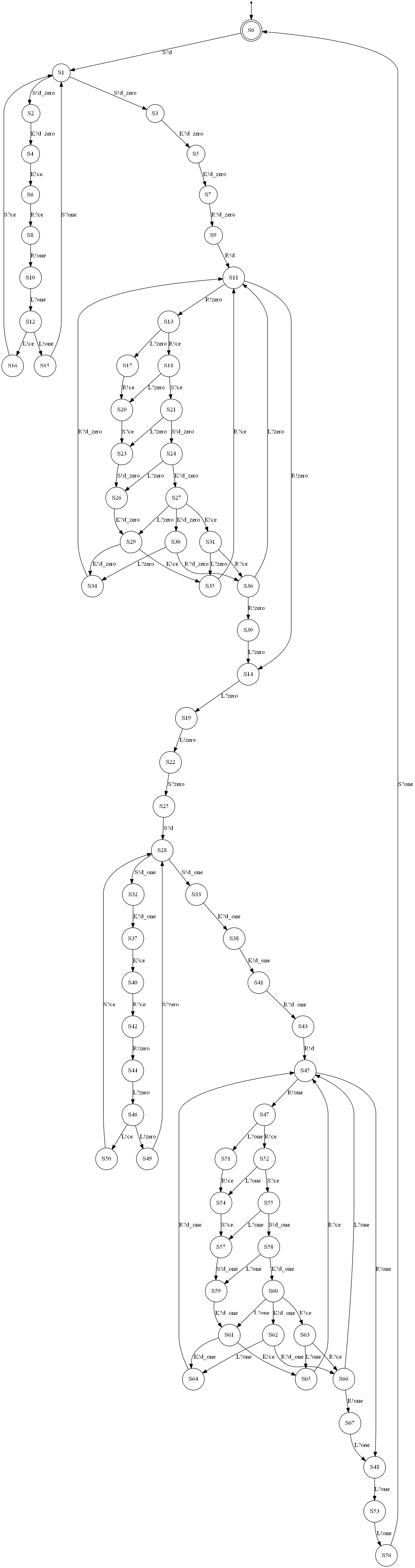}
}
\\
\hline 
\end{tabular}
    
    \caption{Alternating Bit Protocol \cite{high_level_message_sequence_charts} usecase}
    \label{fig:abp_usecase_desc}
\end{figure}

Finally the Alternating Bit Protocol \cite{high_level_message_sequence_charts} usecase is described on Fig.\ref{fig:abp_usecase_desc}.
Fig.\ref{fig:abp_acp} and Fig.\ref{fig:abp_err} resp. present experimental results for the accepted traces and error traces. 

\begin{figure}[h]
    \centering
    \includegraphics[width=.9\textwidth]{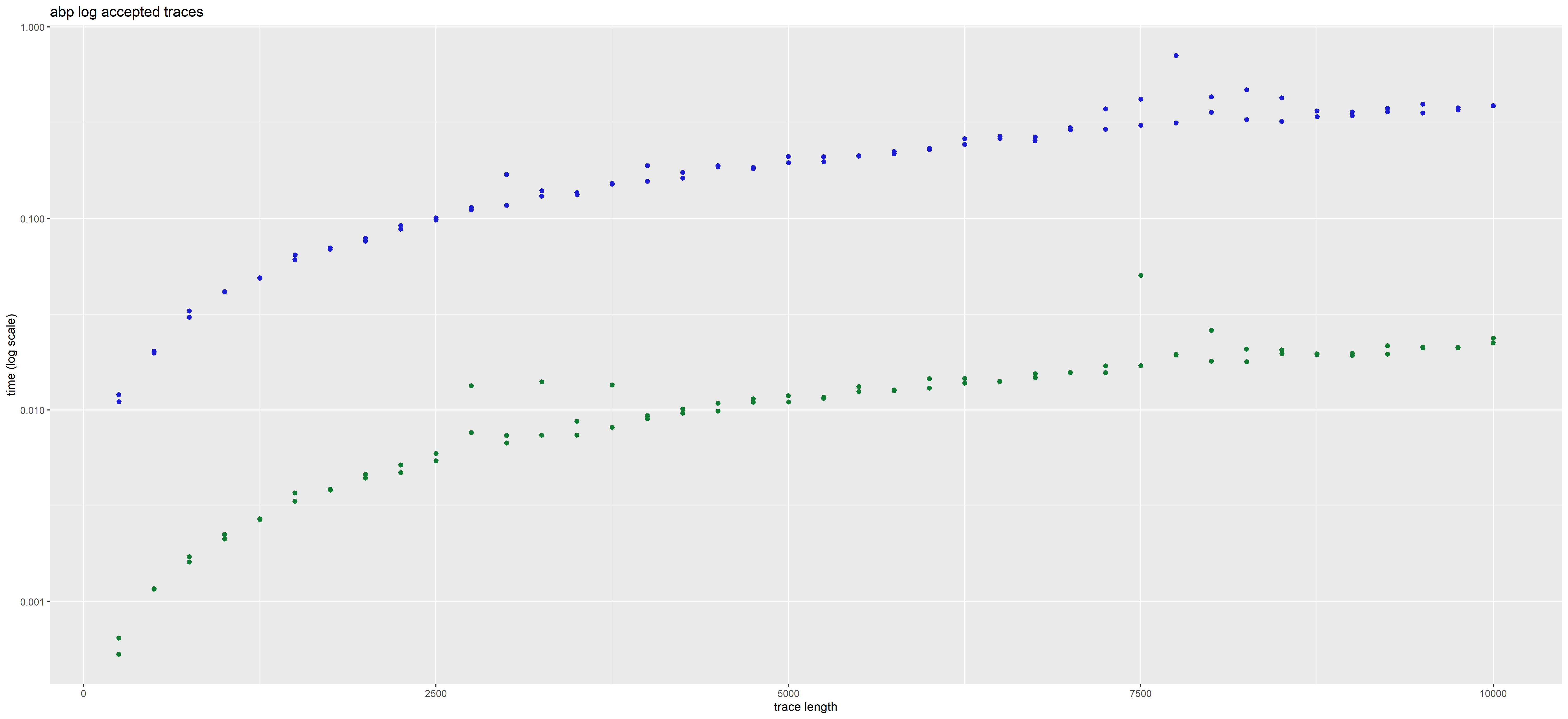}
    \caption{Accepted traces for the ABP \cite{high_level_message_sequence_charts} usecase}
    \label{fig:abp_acp}
\end{figure}

\begin{figure}[h]
    \centering
    \includegraphics[width=.9\textwidth]{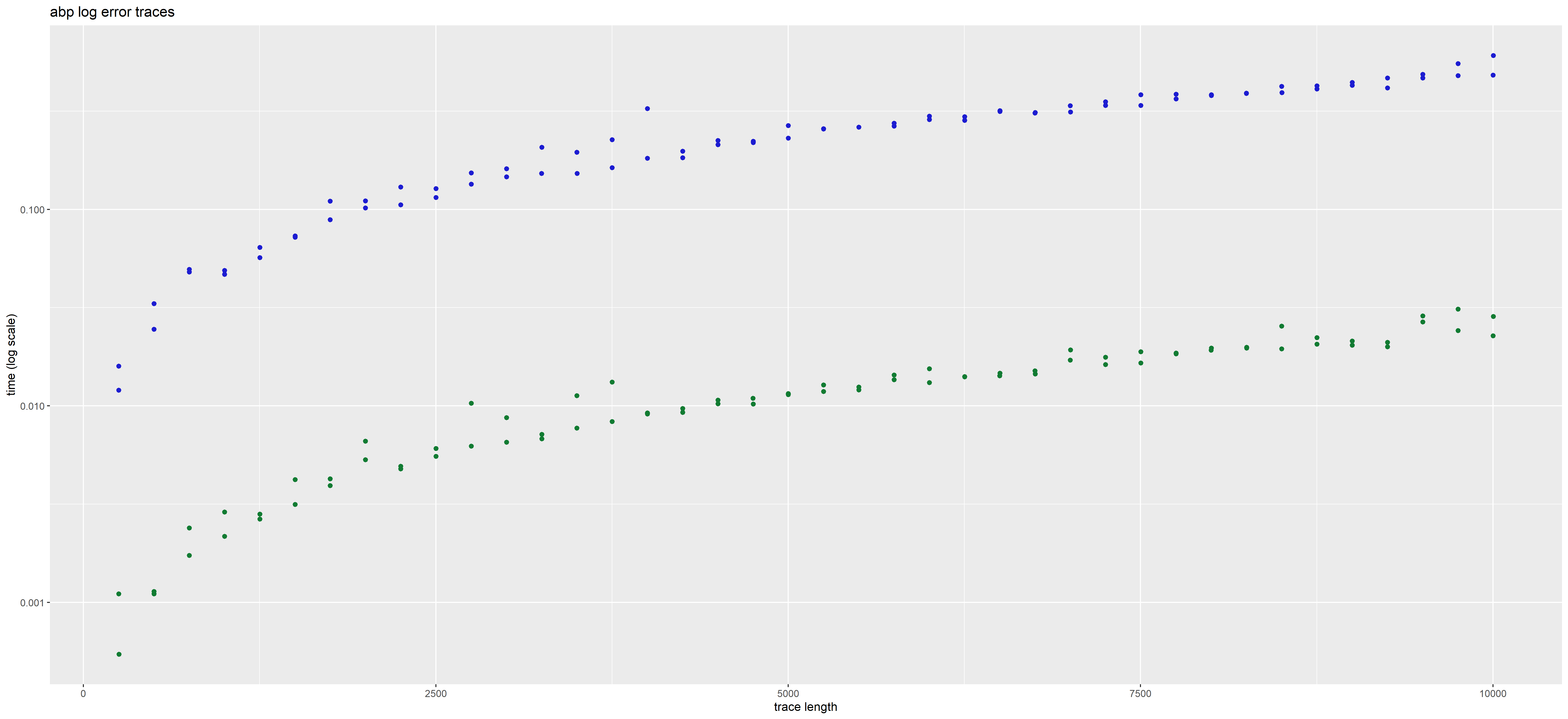}
    \caption{Error traces for the ABP \cite{high_level_message_sequence_charts} usecase}
    \label{fig:abp_err}
\end{figure}

\end{document}